\documentclass{article}

\usepackage{microtype}
\usepackage{graphicx}
\usepackage{subfigure}
\usepackage{booktabs} 
\usepackage{amsmath}
\usepackage{amssymb}
\usepackage{amsthm}

\newcommand{\R}{{\mathbb{R}}}

\newtheorem{theorem}{Theorem}

\newtheorem{lemma}{Lemma}
\newtheorem{proposition}{Proposition}

\newcommand{\E}{{\mathbb{E}}}

\newcommand{\statespace}{\mathcal{X}}
\newcommand{\schedule}{\mathcal{T}_N}

\newcommand{\LSKL}{\mathcal{L}_{\mathrm{SKL}}}

\newcommand{\states}{\textbf{x}}

\newcommand{\nscan}{M}

\newcommand{\Podd}{S_{\mathrm{odd}}}
\newcommand{\Peven}{S_{\mathrm{even}}}

\newcommand{\SKL}{\mathrm{SKL}}

\def\argmax{\operatornamewithlimits{arg\,max}}
\def\argmin{\operatornamewithlimits{arg\,min}}

\newcounter{xxx}
\newcommand{\Var}{\mathrm{Var}}

\newcommand{\ud}{\textrm{d}}

\newcommand\dee{{\mathrm{d}}}

\usepackage{hyperref}

\usepackage[accepted]{icml2021}

\newtheorem{definition}[theorem]{Definition}

\icmltitlerunning{Parallel tempering on optimized paths}

\begin{document}

\twocolumn[
\icmltitle{Parallel tempering on optimized paths}



\icmlsetsymbol{equal}{*}

\begin{icmlauthorlist}
\icmlauthor{Saifuddin Syed}{equal,UBC}
\icmlauthor{Vittorio Romaniello}{equal,UBC}
\icmlauthor{Trevor Campbell}{UBC}
\icmlauthor{Alexandre Bouchard-C\^{o}t\'{e}}{UBC}
\end{icmlauthorlist}

\icmlaffiliation{UBC}{Department of Statistics, University of British Columbia, Vancouver, Canada}

\icmlcorrespondingauthor{Saifuddin Syed}{saif.syed@stat.ubc.ca}
\icmlcorrespondingauthor{Vittorio Romaniello}{vittorio.romaniello@stat.ubc.ca}

\icmlkeywords{Parallel tempering, bayesian inference, MCMC, information geometry}

\vskip 0.3in
]



\printAffiliationsAndNotice{\icmlEqualContribution} 

\begin{abstract}
Parallel tempering (PT) is a class of Markov chain Monte Carlo algorithms 
that constructs a path of distributions annealing between a tractable
reference and an intractable target, and then interchanges states along the path to improve 
mixing in the target. 
The performance of PT depends on how quickly a sample from the reference
distribution makes its way to the target, which in turn depends on the particular
path of annealing distributions.
However, past work on PT has used only simple paths constructed from convex combinations
of the reference and target log-densities.
This paper begins by demonstrating that this path performs poorly
in the setting where the reference and target are nearly mutually singular.
To address this issue, we expand the framework of PT to general families of paths, 
formulate the choice of path as an optimization problem that admits tractable gradient estimates,
and propose a flexible new family of spline interpolation paths for use
in practice.
Theoretical and empirical results both demonstrate that our proposed methodology breaks 
previously-established upper performance limits for traditional paths.

\end{abstract}

\section{Introduction}
Markov Chain Monte Carlo (MCMC) methods are widely used to approximate
intractable expectations with respect to un-normalized probability
distributions over general state spaces. For hard problems, MCMC can suffer
from poor mixing. For example, faced with well-separated modes, MCMC methods often get
trapped exploring local regions of high probability. Parallel tempering (PT)
is a widely applicable methodology \citep{geyer1991markov} to tackle
poor mixing of MCMC algorithms. 

Suppose we seek to approximate an expectation with respect to an intractable
\emph{target density} $\pi_1$.
Denote by $\pi_0$ a \emph{reference density} defined on the same
space, which is assumed to be tractable in
the sense of the availability of an efficient sampler.
This work is motivated by the case where
$\pi_0$ and $\pi_1$ are nearly mutually singular.
A typical case is where the target is a Bayesian posterior distribution, 
the reference is the prior---for which i.i.d.~sampling is typically possible---and
the prior is misspecified.

PT methods are based on a specific continuum of densities 
$\pi_t\propto \pi_0^{1-t}\pi_1^t$, $t\in[0,1]$,
bridging $\pi_0$
and $\pi_1$.  This path of intermediate
distributions is known as the power posterior path in the literature, but  in
our framework it will be more natural to think of these continua as a
\emph{linear paths}, as they linearly interpolate between log-densities. PT
algorithms discretize the path at some $0=t_0<\dots<t_N=1$ to obtain a
sequence of densities $\pi_{t_0}, \pi_{t_1}, \dots, \pi_{t_N}$. See 
Figure~\ref{fig:gaussian_paths} (top) for an example of a linear path for two nearly mutually singular Gaussian
distributions.

Given the path discretization, PT involves running $N+1$ MCMC chains 
that together target the product distribution
$\pi_{t_0} \pi_{t_1} \cdots \pi_{t_N}$. Based on the assumption that the chain
$\pi_0$ can be sampled efficiently, PT uses swap-based interactions between
neighbouring chains to propagate the exploration done in $\pi_0$ into improved
exploration in the chain of interest $\pi_1$. By designing these swaps as Metropolis--Hastings
moves, PT guarantees that the marginal distribution of the $N^\text{th}$ chain
converges to $\pi_1$; and in practice, the rate of convergence is often much faster
compared to running a single chain \citep{woodard2009rapid}. PT algorithms are
extensively used in hard sampling problems arising in statistics,
physics, computational chemistry, phylogenetics, and machine learning
\citep{desjardins2014deep,ballnus2017comprehensive,kamberaj2020molecular,mueller2020adaptive}.

Notwithstanding empirical and theoretical successes, existing PT algorithms
also have well-understood theoretical limitations. Earlier work focusing on
the theoretical analysis of reversible variants of PT has shown that adding too
many intermediate chains can actually deteriorate performance
\citep{lingenheil2009efficiency, atchade_towards_2011}.  Recent work
\citep{syed_non_reversible_2019} has shown that a nonreversible variant of PT
\citep{okabe2001replica} is guaranteed to dominate its classical reversible
counterpart, and moreover that in the nonreversible regime adding more chains
does not lead to performance collapse.  However, even with these more efficient
non reversible PT algorithms, \citet{syed_non_reversible_2019} established that
the improvement brought by higher parallelism will asymptote to a fundamental
limit known as the \emph{global communication barrier}. 
	
In this work, we show that by generalizing the class of paths interpolating
between $\pi_0$ and $\pi_1$ from linear to nonlinear, the global communication
barrier can be broken, leading to substantial performance improvements.
Importantly, the nonlinear path used to demonstrate this breakage is computed
using a practical algorithm that can be used in any situation where PT is
applicable. An example of a path optimized using our algorithm is shown in
Figure~\ref{fig:gaussian_paths} (bottom). 

We also present a detailed theoretical analysis of parallel tempering algorithms
based on nonlinear paths.  Using this analysis we prove that the performance gains
obtained by going from linear to nonlinear path PT algorithms can be arbitrarily large. 
Our theoretical analysis also motivates a principled
objective function used to optimize over a parametric family of paths.


\textbf{Literature review}\hspace{.1cm} Beyond parallel tempering, several methods to
approximate intractable integrals 
rely on a path of distributions from a reference to a target distribution, 
and there is a rich literature on the construction and optimization of 
nonlinear paths for annealed importance sampling type algorithms 
\citep{gelman1998simulating,rischard2018unbiased,grosse2013annealing, brekelmans2020annealed}. 
These algorithms are highly parallel; however, for challenging problems, even when 
combined with adaptive step size procedures \citep{Zhou2016} they typically 
suffer from particle degeneracy \citep[Sec. 7.4]{syed_non_reversible_2019}. 
Moreover, these methods use different path optimization criteria which are 
not well motivated in the context of parallel tempering. 

Some special cases of non-linear paths have been used in the PT literature
\cite{whitfield2002generalized, tawn2020weight}.
\citet{whitfield2002generalized} construct a non-linear path inspired by the
concept of Tsallis entropy, a generalization of Boltzmann-Gibbs entropy, but do
not provide algorithms to optimize over this path family.  The work of
\citet{tawn2020weight}, also considers a specific example of a nonlinear path
distinct from the ones explored in this paper.  However, the construction of
the nonlinear path in \citet{tawn2020weight} requires knowledge of the location
of the modes of $\pi_1$ and hence makes their algorithm less broadly applicable
than standard PT.

\section{Background}\label{sec:NRPT}

In this section, we provide a brief overview of 
parallel tempering (PT) \cite{geyer1991markov},
as well as recent results on nonreversible communication \cite{okabe2001replica,sakai_irreversible_2016,syed_non_reversible_2019}.
Define a reference unnormalized density function $\pi_0$ for which sampling is tractable,
and an unnormalized target density function $\pi_1$ for which sampling is intractable; 
the goal is to obtain samples from $\pi_1$.\footnote{We assume all distributions share a common state space $\statespace$ throughout,
and will often suppress the arguments of (log-)density functions---i.e., $\pi_1$ instead of $\pi_1(x)$---for notational brevity.} Define a path of distributions
$\pi_t \propto \pi_0^{1-t}\pi^t$ for $t\in[0, 1]$ from the reference to the target. 
Finally, define the annealing schedule $\schedule$ to be a monotone sequence in $[0, 1]$, satisfying
\begin{align*}
 \schedule&=(t_n)_{n=0}^N, \quad 0 = t_0 \leq t_1 \leq \dots \leq t_N = 1\\
 \|\schedule\| &= \max_{n\in \{0, \dots, N-1\}} t_{n+1}-t_n.
\end{align*}
The core idea of parallel tempering is to construct a Markov chain
$(X^0_m,\dots,X^N_m)$, $m=1, 2, \dots$ that (1) has invariant distribution
$\pi_{t_0}\cdot\pi_{t_1}\cdots\pi_{t_N}$---such that we can treat 
the marginal chain $X^N_n$ as samples from the target $\pi_1$---and (2) swaps components
of the state vector such that independent samples from component 0 (i.e., the reference $\pi_0$) traverse along the annealing
path and aid mixing in component $N$ (i.e., the target $\pi_1$).
This is possible to achieve by iteratively performing 
a \emph{local exploration} move followed by a \emph{communication} move as shown in Algorithm \ref{alg:NRPT}.

\paragraph{Local Exploration}
Given $(X^0_{m-1},\dots,X^N_{m-1})$, we obtain an intermediate state
$(\tilde{X}^0_m,\dots,\tilde{X}^N_m)$ by updating the $n^\text{th}$
component using any MCMC move targeting $\pi_{t_n}$, for $n=0, \dots, N$. 
This move can be performed in 
parallel across components since each is updated independently. 

\paragraph{Communication} Given the intermediate state $(\tilde{X}^0_m,\dots,\tilde{X}^N_m)$,
we apply pairwise swaps of components
$n$ and $n+1$, $n\in S_m\subset\{0,\dots,N-1\}$ for swapped index set $S_m$.
Formally, a swap is a move from
$(x^0,\dots,x^N)$ to $(x^0,\dots,x^{n+1},x^n,\dots,x^N)$,
which is accepted with probability 
\begin{align}
\alpha_n = 1\wedge \frac
{\pi_{t_n}(x^{n+1})\pi_{t_{n+1}}(x^n)}
{\pi_{t_n}(x^{n})\pi_{t_{n+1}}(x^{n+1})}.
\end{align}
Since each swap only depends on components $n$, $n+1$, one can perform
all of the swaps in $S_m$ in parallel, as long as $n \in S_m$ implies $(n+1) \notin S_m$.
The largest collection of such non-interfering swaps
is $S_m\in \{\Peven,\Podd\}$, where
$\Peven,\Podd$ are the even and odd subsets of $\{0,\dots,N-1\}$ respectively.
In \emph{non-reversible PT} (NRPT) \citep{okabe2001replica}, the swap set $S_m$ at each step $m$
is set to 
\[
S_m = \left\{\begin{array}{ll}
\Peven & \text{if }m \text{ is even}\\
\Podd &  \text{if }m \text{ is odd.}\\
\end{array}\right.
\]

\paragraph{Round trips}
The performance of PT is sensitive to both the local exploration and
communication moves. The quantity commonly used to evaluate the performance of
MCMC algorithms is the effective sample size (ESS); however, ESS measures the combined
performance of local exploration and communication, and is not 
able to distinguish between the two. Since the major difference between PT and standard MCMC
is the presence of a communication step, we require a way to measure communication performance
in isolation such that we can compare PT methods
without dependence on the details of the local exploration move. The \emph{round trip rate} is a
performance measure from the PT literature
\cite{katzgraber2006feedback,lingenheil2009efficiency}
that is designed to assess communication efficiency alone.
We say a \emph{round trip} has occurred when a new sample from the reference
$\pi_0$ travels to the target $\pi_1$ and then back to $\pi_0$; the {round trip rate} 
$\tau(\schedule)$ is the frequency at which round trips occur. Based
on simplifying assumptions on the local exploration moves,
the round trip rate $\tau(\schedule)$  may be expressed as \citep[Section~3.5]{syed_non_reversible_2019}
\begin{align}
\tau(\schedule) = \left(2+2\sum_{n=0}^{N-1} \frac{r(t_n,t_{n+1})}{1-r(t_n,t_{n+1})}\right)^{-1}, \label{eq:nrpt_rtr}
\end{align}
 where $r(t,t')$ is the expected probability of rejection between chains $t,t' \in [0,1]$,
\[
r(t,t')=\E\left[1\wedge \frac
{\pi_{t}(X')\pi_{t'}(X)}
{\pi_{t}(X)\pi_{t'}(X')}\right],
\]
and $X, X'$ have distributions $X\sim\pi_t$, $X'\sim\pi_{t'}$. 
Further, if $\schedule$ is refined so that $\|\schedule\|\to 0$ as $N\to \infty$,
we find that the asymptotic (in $N$) round trip rate is
\begin{align}
 \tau_\infty = \lim_{N\to\infty} \tau(\schedule) = \left(2+2\Lambda\right)^{-1}, \label{eq:nrpt_artr}
\end{align}
where $\Lambda \geq 0$ is a constant associated with the pair $\pi_0, \pi_1$ called the \emph{global communication barrier} \cite{syed_non_reversible_2019}. Note that $\Lambda$ does not depend on the number of chains $N$ or discretization schedule $\schedule$.

\begin{algorithm}[tb]
   \caption{NRPT}
   \label{alg:NRPT}
\begin{algorithmic}
   \REQUIRE state $\states_0$, path $\pi_t$, schedule $\schedule$, \# iterations $\nscan$
   \STATE $r_{n} \gets 0$ for all $n\in\{0,\dots, N-1\}$
   \FOR{$m=1$ {\bfseries to} $\nscan$}
   \STATE $\tilde\states_m\gets \texttt{LocalExploration}(\states_{m-1})$
   \STATE $S_m\gets \Peven$ if $m$ is even, otherwise $S_m\gets \Podd$
   \FOR{$n=0$ {\bfseries to} $N-1$}
   \STATE $\alpha_n \gets 1 \wedge \frac
{\pi_{t_n}(x^{n+1})\pi_{t_{n+1}}(x^n)}
{\pi_{t_n}(x^{n})\pi_{t_{n+1}}(x^{n+1})}$
   \STATE $r_n \gets r_n + (1-\alpha_n)$
   \STATE $U_n\sim \mathrm{Unif}(0,1)$
   \IF{$n\in S_m$ \AND $U_n \leq \alpha_n$}
   \STATE $(\tilde x_m^n,\tilde x_m^{n+1})\gets (\tilde x_m^{n+1},\tilde x_m^n)$
   \ENDIF
   \STATE $\states_m\gets \tilde\states_m$
   \ENDFOR
   \ENDFOR
   \STATE $r_n\gets r_n/\nscan$ for $n\in\{0, \dots, N-1\}$
   \STATE {\bfseries Return:}  $\{\states_m\}_{m=1}^{\nscan}, \{r_n\}_{n=0}^{N-1}$
\end{algorithmic}
\end{algorithm}

\section{General annealing paths}\label{sec:general_annealing}

In the following, we use the terminology \emph{annealing path} to describe a
continuum of distributions interpolating between $\pi_0$ and $\pi_1$; this
definition will be formalized shortly. The previous work reviewed in the last
section assumes that the annealing path has the form $\pi_t\propto
\pi_0^{1-t}\pi_1^t$, i.e., that the annealing path linearly interpolates
between the log densities. A natural question is whether using other
paths could lead to an improved round trip rate. 

In this work we show that the answer to this question is positive. The following 
proposition demonstrates that the traditional path
 $\pi_t \propto \pi_0^{1-t}\pi_1^t$ suffers from an arbitrarily suboptimal global communication barrier
even in simple examples with Gaussian reference and target distributions.
\begin{proposition}\label{prop:badlinear}
Suppose the reference and target distributions are $\pi_0 = \mathcal{N}(\mu_0, \sigma^2)$ and $\pi_1 = \mathcal{N}(\mu_1, \sigma^2)$, 
and define $z=|\mu_1-\mu_0|/\sigma$. Then as $z\to\infty$,
\begin{enumerate}
\item the path $\pi_t\propto \pi_0^{1-t}\pi_1^t$ has $\tau_\infty = \Theta(1/z)$, and
\item there exists a path of Gaussians distributions with $\tau_\infty = \Omega( 1/\log z )$.
\end{enumerate}
\end{proposition}
\begin{figure}[t!]
\vskip 0.2in
\begin{center}
\centerline{\includegraphics[width=\columnwidth]{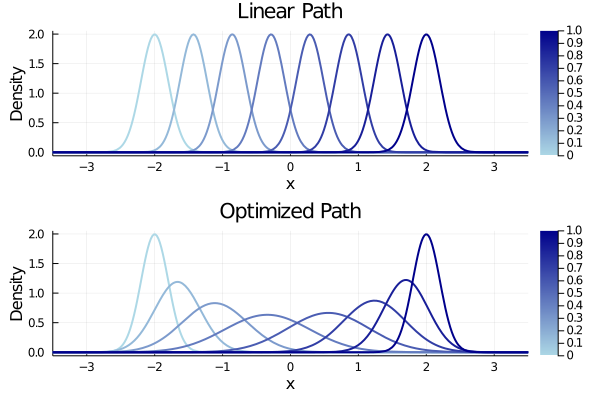}}
\caption{Two annealing paths between a $\pi_0=N(-2,0.2^2)$ (light blue) and $\pi_1=N(2,0.2^2)$ (dark blue) :
the traditional linear path (top) and an optimized nonlinear path (bottom). While the distributions in the linear path are nearly mutually singular, those in the optimized path overlap substantially, leading to faster round trips.}
\label{fig:gaussian_paths}
\end{center}
\vskip -0.2in
\end{figure}

Therefore, upon decreasing the variance of the reference and target while holding their means fixed,
the traditional linear annealing path obtains an exponentially smaller asymptotic round trip rate
than the optimal path of Gaussian distributions.  Figure \ref{fig:gaussian_paths}
provides an intuitive explanation. The standard path (top) corresponds to
a set of Gaussian distributions with mean interpolated between the reference and target.
If one reduces the variance of the reference and target, so does the 
variance of the distributions along the path.
For any fixed $N$, these distributions become nearly mutually singular, leading to
arbitrarily low round trip rates.
The solution to this issue (bottom) is to allow the distributions along the
path to have increased variances, thereby maintaining mutual overlap and the
ability to swap components with a reasonable probability. 
This motivates the need to design more general annealing paths.
 In the following, we introduce the precise
general definition of an annealing path, an analysis of path communication 
efficiency in parallel tempering, and a rigorous formulation of---and solution to---the
problem of tuning path parameters to maximize communication efficiency. 

\subsection{Assumptions}\label{sec:assumptions}

Let $\mathcal{P}(\statespace)$ be the set of probability densities with full 
support on a state space $\statespace$.
For any collection of densities $\pi_t \in \mathcal{P}(\statespace)$ with index $t \in [0, 1]$, associate
to each a log-density function $W_t$ such that
\begin{align}\label{eq:non-linear-path}
\pi_{ t}(x) = \frac{1}{Z_{ t}} \exp\left(W_{t}(x)\right),  \quad x\in\statespace,
\end{align}
where $Z_{t}=\int_\statespace \exp\left(W_{t}(x)\right)\ud x$ is the normalizing constant. 
Definition \ref{defn:path} provides the conditions necessary to form a
path of distributions from a reference $\pi_0$ to a target $\pi_1$ that are
well-behaved for use in parallel tempering.

\begin{definition}\label{defn:path}
An \emph{annealing path} is a map $\pi_{(\cdot)}:[0, 1] \to \mathcal{P}(\statespace)$, 
denoted $t \mapsto \pi_t$, such that for all $x\in\statespace$,  $\pi_t(x)$ is continuous in $t$.
\end{definition}

There are many ways to move beyond the standard linear path $\pi_t \propto \pi_0^{1-t}\pi_1^t$. For example, consider a nonlinear path $\pi_t \propto
\pi_0^{\eta_0(t)}\pi_1^{\eta_1(t)}$ where $\eta_i : [0, 1] \to \R$ are continuous
functions such that $\eta_0(0) = \eta_1(1) = 1$ and $\eta_0(1) = \eta_1(0) =
0$. As long as for all $t\in[0,1]$, $\pi_t$ is a normalizable density this is a valid annealing path between $\pi_0$ and $\pi_1$.
Further, note that the path parameter does not necessarily have to appear as an exponent: consider for example the mixture path  $\pi_{t} \propto (1-t)\pi_0 + t \pi$.
Section~\ref{sec:expfamspline} provides a more detailed example based on linear splines. 

\subsection{Communication efficiency analysis}\label{sec:communication-analysis}

Given a particular annealing path satisfying Definition \ref{defn:path}, we
require a method to characterize the round trip rate performance of parallel
tempering based on that path. The results presented in this section form the basis of the objective
function used to optimize over paths, as well as the foundation for the proof
of  Proposition~\ref{prop:badlinear}. 

We start with some notation for the rejection rates involved when Algorithm~1
is used with nonlinear paths (Equation~\ref{eq:non-linear-path}). 
For $t, t' \in [0,1]$, $x\in\statespace$, define the rejection rate function
$r:[0,1]^2\to[0,1]$ to be
\begin{align*}
&r(t, t') = 1 - \E\left[ \exp\left(\min\{0,A_{t, t'}(X, X')\}\right)\right] \\
&A_{t, t'}(x,x') = (W_{t'}(x)-W_t(x)) - (W_{t'}(x') - W_t(x')),
\end{align*}
where $X\sim \pi_{t}$ and $X'\sim\pi_{t'}$ are independent. 
Assuming that all chains have reached stationarity,
and all chains undergo \emph{efficient local exploration}---i.e., 
$W_{t}(X), W_{t}(\tilde{X})$ are independent when $X\sim \pi_t$ and $\tilde{X}$ is 
generated by local exploration from $X$---then
the round trip rate for a particular schedule $\schedule$
has the form given earlier in Equation (\ref{eq:nrpt_rtr}). 
This statement follows from \citet[Cor.~1]{syed_non_reversible_2019} without modification,
because the proof of that result does not depend on the form of the acceptance ratio.

Our next objective is to characterize the asymptotic communication efficiency of a nonlinear path in the regime 
where $N\to\infty$ and $\|\schedule\|\to 0$---which 
establishes its fundamental ability to take advantage of parallel computation. 
In other words, we require a generalization of the asymptotic result in Equation~(\ref{eq:nrpt_artr}). 
In this case, previous theory relies on the particular form of the acceptance ratio for linear paths \cite{syed_non_reversible_2019}; in the remainder of this section, we provide a generalization of the asymptotic result for nonlinear paths
by piecewise approximation by a linear spline.

For $t, t'\in[0,1]$, 
define $\Lambda(t,t')$ to be the global communication barrier for 
the linear secant connecting $\pi_t$ and $\pi_t'$,
\begin{align}\label{eq:def_GCB_linear}
\Lambda(t,t')&=\frac{1}{2} \int_0^1 \E[|A_{t,t'}(X_s,X'_s)|]\dee s,\\
X_s, X'_s &\overset{\text{i.i.d.}}{\sim} \frac{1}{Z_s(t,t')}\exp\left((1-s)W_t+sW_{t'}\right).\notag
\end{align}
Lemma \ref{lemma:nonlinear_gcb_err} shows that the global communication barrier $\Lambda(t,t')$ along the secant
of the path connecting $t$ to $t'$ is a good approximation of the true rejection rate $r(t,t')$ with $O(|t-t'|^3)$
error as $t\to t'$.
\begin{lemma}\label{lemma:nonlinear_gcb_err} 
Suppose that for all $x\in\statespace$, $W_t(x)$ is piecewise continuously differentiable in $t$,
 and that there exists $V_1:\statespace\to [0,\infty)$ such that 
\begin{equation}\label{eq:V1_bound}
\forall x \in \statespace, \,\, \sup_{t\in[0,1]}\left|\frac{\dee W_t}{\dee t}(x)\right|\leq V_1(x),
\end{equation}
and 
\begin{equation}\label{eq:V1_integrability}
\sup_{t\in[0,1]}\E_{\pi_t}[V_1^3]<\infty.
\end{equation}
Then there exists a constant $C < \infty$ independent of $t,t'$ such that for all $t, t'\in[0,1]$,
\begin{equation}
|r(t,t') - \Lambda(t,t')| \leq  C|t-t'|^3.
\end{equation}
\end{lemma} 
A direct consequence of Lemma \ref{lemma:nonlinear_gcb_err} is that for any fixed schedule $\schedule$,
\begin{equation}
\left|\sum_{n=0}^{N-1} r(t_{n},t_{n+1}) - \Lambda(\schedule)\right| \leq  C\|\schedule\|^2,
\end{equation}
where $\Lambda(\schedule)  = \sum_{n=0}^{N-1}\Lambda(t_{n},t_{n+1})$. 
Intuitively, in the $\|\schedule\|\approx 0$ regime where
rejection rates are low,
\[
\frac{r(t_n,t_{n+1})}{1-r(t_n,t_{n+1})} \approx r(t_n,t_{n+1}),
\]
and we have  that 
$\tau(\schedule) \approx (2+2\Lambda(\schedule))^{-1}$.
Therefore, $\Lambda(\schedule)$ characterizes the communication efficiency of the path
in a way that naturally extends the global communication barrier from the linear path case.
Theorem \ref{thm:general_GCB} provides the precise statement: the convergence is uniform
in $\schedule$ and depends only on $\|\schedule\|$, and $\Lambda(\schedule)$ 
itself converges to a constant $\Lambda$ in the asymptotic regime. 
We refer to $\Lambda$,
defined below in Equation (\ref{eq:nonlineargcb}) as the global communication barrier for the general annealing path. 
\begin{theorem}\label{thm:general_GCB}
Suppose that for all $x\in\statespace$, $W_t(x)$ is piecewise
twice continuously differentiable in $t$, that
there exists $V_1:\statespace\to [0,\infty)$ satisfying
\eqref{eq:V1_bound} and \eqref{eq:V1_integrability},
and
that there exists $V_2:\statespace\to [0,\infty)$, $\epsilon > 0$ satisfying
\begin{equation}\label{eq:second_deriv}
 \forall x\in\statespace, \,\, \sup_{t\in[0,1]}\left|\frac{\dee^2 W_t}{\dee t^2}(x)\right|\leq V_2(x),
\end{equation} 
 and 
\begin{equation}\label{eq:mgf}
\sup_{t\in[0,1]} \E_{\pi_t}\left[(1+V_1)\exp(\epsilon V_2)\right] < \infty.
\end{equation}
Then,
\begin{align}
\lim_{\delta\to 0} \sup_{\schedule : \|\schedule\| \leq \delta} \left|(2+2\Lambda(\schedule))^{-1} - \tau(\schedule)\right| &= 0\label{eq:nonlineartauconv}\\
\text{and}\quad \lim_{\delta\to 0} \sup_{\schedule : \|\schedule\| \leq \delta} \left|\Lambda(\schedule) - \Lambda\right| &= 0,\label{eq:nonlineargcb}
\end{align}
where $\Lambda = \int_0^1\lambda(t)\dee t$ for an instantaneous rejection rate function $\lambda : [0,1] \to [0,\infty)$ given by
\begin{align*}
\lambda(t)&=\lim_{\Delta t\to 0}\frac{r(t+\Delta t,t)}{|\Delta t|}\\
&=\frac{1}{2}\E\left[\left|\frac{dW_{t}}{dt}(X_t)-\frac{dW_{t}}{dt}(X_t')\right|\right], \quad X_t,X'_t\overset{\text{i.i.d.}}{\sim}\pi_t. 
\end{align*}
\end{theorem}
The integrability condition is required to control the tail behaviour of distributions
formed by linearized approximations to the path $\pi_t$.
This condition is satisfied by a wide range of annealing paths,
e.g., the linear spline paths proposed in this work---since 
in that case $V_2 = 0$.

\subsection{Annealing path families and optimization}\label{sec:path_tuning}
It is often the case that there are a set of candidate annealing paths in consideration 
for a particular target $\pi$.
For example, if a path has tunable parameters $\phi\in\Phi$ that govern its shape, we
can generate a collection of annealing paths that all target $\pi$ by varying the parameter $\phi$.
We call such collections an annealing path family.

\begin{definition}\label{defn:pathfamily}
An \emph{annealing path family} for target $\pi_1$
is a collection of annealing paths $\{\pi^\phi_t\}_{\phi\in\Phi}$ such that
for all parameters $\phi \in \Phi$, $\pi^\phi_1=\pi_1$.
\end{definition}
There are many ways to construct useful annealing path families.
For example, if one is provided a parametric family of variational distributions $\{q_{\phi} : \phi\in\Phi\}$
for some parameter space $\Phi$,
one can construct the annealing path family of linear paths
$\pi^\phi_t=q_\phi^{1-t}\pi_1^t$ from a variational reference $q_\phi$ to the
target $\pi_1$. More generally, given $\eta_i(t)$ satisfying the constraints in
Section \ref{sec:assumptions}, $\pi^\phi_t=q_\phi^{\eta_0(t)}\pi_1^{\eta_1(t)}$
defines a nonlinear annealing path family. 
Another example of an annealing path family used in the context of PT are $q$-paths $\{\pi_t^q\}_{q\in[0,1]}$
\citep{whitfield2002generalized}. Given a fixed reference and target $\pi_0,\pi_1$, the path
$\pi_t^q$ interpolates between the mixture path ($q=0$) and the linear path ($q=1$) 
\citep{brekelmans2020annealed}. In Section \ref{sec:expfamspline}, we provide a new
flexible class of nonlinear paths based on splines that is designed specifically
to enhance the performance of parallel tempering.

Since every path in an annealing path family has the desired target distribution $\pi_1$, 
we are free to optimize the path over the tuning parameter space $\phi\in\Phi$
in addition to optimizing the schedule $\schedule$.\footnote{We assume that the optimization over $\phi$ ends after a finite number of iterations to sidestep the potential pitfalls of adaptive MCMC methods \cite{andrieu_ergodicity_2006}.} 
Motivated by the analysis of Section~\ref{sec:communication-analysis}, a natural objective function
for this optimization to consider is the non-asymptotic round trip rate
\begin{align}
 \phi^\star, \schedule^\star &= \argmax_{\phi\in\Phi,\schedule} \,\tau^\phi(\schedule)\label{eq:tauopt}\\
&= \argmin_{\phi\in\Phi,\schedule} \sum_{n=0}^{N-1} \frac{r^\phi(t_{n}, t_{n+1})}{1-r^\phi(t_n,t_{n+1})},\label{eq:nonasympcb}
\end{align}
where now the round trip rate and rejection rates depend both on the schedule and path parameter, 
denoted by superscript $\phi$.
We solve this optimization using an approximate coordinate-descent procedure, iterating between an update of the schedule $\schedule$ for a fixed path parameter $\phi\in\Phi$, followed by a gradient step in $\phi$ based on a surrogate objective function and a fixed schedule. This is summarized in Algorithm \ref{alg:PT_tuning}. We outline the details of schedule and path tuning procedure in the following. 

\begin{algorithm}[t!]
   \caption{PathOptNRPT}
   \label{alg:PT_tuning}
\begin{algorithmic}
   \REQUIRE state $\states$, path family $\pi^\phi_t$, parameter $\phi$, \# chains $N$, \# PT iterations $\nscan$, \# tuning steps $S$, learning rate $\gamma$
   \STATE $\schedule \gets (0, 1/N, 2/N, \dots, 1)$
   \FOR{$s=1$ {\bfseries to} $S$}
   \STATE $\{\states_m\}_{m=1}^M, (r_n)_{n=0}^N\gets \texttt{NRPT}(\states,\pi^\phi_t,\schedule,M)$
   \STATE $\lambda^\phi \gets \texttt{CommunicationBarrier}(\schedule,\{r_n\})$
   \STATE $\schedule\gets \texttt{UpdateSchedule}(\lambda^\phi,N)$
   \STATE $\phi \gets \phi - \gamma \nabla_\phi \sum_{n=0}^{N-1}\SKL(\pi^\phi_{t_n}, \pi^\phi_{t_{n+1}})$
   \STATE $\states \gets \states_{M}$
   \ENDFOR
    \STATE {\bfseries Return:} $\phi,\schedule$
\end{algorithmic}
\end{algorithm}

\paragraph{Tuning the schedule} 
Fix the value of $\phi$, which fixes the path. We adapt a schedule tuning algorithm  
from past work to update the schedule $\schedule=(t_n)_{n=0}^N$ \citep[Section~5.1]{syed_non_reversible_2019}.
Based on the same argument as this previous work, we obtain that when $\|\schedule\|\approx 0$, the non-asymptotic round trip rate is maximized
when the rejection rates are all equal. The schedule that approximately
achieves this satisfies 
\begin{align}
\forall n \in \{1, \dots, N-1\}, \quad \frac{1}{\Lambda^{\phi}}\int_0^{t_n} \lambda^\phi(s)\dee s = \frac{n}{N}. \label{eq:equicrit}
\end{align}
Following \citet{syed_non_reversible_2019}, 
we use Monte Carlo estimates of the rejection rates $r^\phi(t_n, t_{n+1})$ to approximate $t \mapsto \int_0^t \lambda^\phi(s) \ud s$, $s\in[0,1]$ via a monotone cubic spline,  
and then use bisection search to solve for each $t_n$ according to Equation (\ref{eq:equicrit}).

\paragraph{Optimizing the path} Fix the schedule $\schedule$; we now want to
improve the path itself by modifying $\phi$.  However, in
challenging problems this is not as simple as taking a gradient step for the
objective in Equation \eqref{eq:tauopt}.  In particular, in early
iterations---when the path is near its oft-poor initialization---the rejection rates
satisfy $r^\phi(t_n,t_{n+1}) \approx 1$.  As demonstrated empirically in Appendix
\ref{sec:snr}, gradient estimates in this regime exhibit a low signal-to-noise
ratio that precludes their use for optimization.

We propose a surrogate, the symmetric KL divergence, motivated as follows. Consider first the global communication barrier $\Lambda^\phi(\schedule)$ for the linear spline approximation to the path;
Theorem \ref{thm:general_GCB} guarantees that as long as $\|\schedule\|$ is small enough,
one can optimize $\Lambda^\phi(\schedule)$ in place of the round trip rate $\tau^\phi(\schedule)$.
By Jensen's inequality,
\begin{equation*}
\frac{1}{N^2}\Lambda^\phi(\schedule)^2\leq\frac{1}{N}\sum_{n=0}^{N-1}\Lambda^\phi(t_{n},t_{n+1})^2.
\end{equation*}
Next, we apply Jensen's inequality again to the definition of $\Lambda^\phi(t_{n},t_{n+1})$ from \eqref{eq:def_GCB_linear}, 
which shows that
\begin{equation*}
\Lambda^\phi(t_{n},t_{n+1}))^2\leq\frac{1}{4}\int_0^1\E[A_{t_{n},t_{n+1}}(X_s,X'_s)^2]\dee s,
\end{equation*}
where $X_s$ are defined in \eqref{eq:def_GCB_linear} are drawn from the linear path between $\pi^\phi_{t_{n}}$ and $\pi^\phi_{t_{n+1}}$. Finally, we note that the inner expectation is the path integral of the Fisher information metric along the linear path and evaluates to the symmetric KL
divergence \citep[Result 4]{dabak2002relations},
\begin{align*}
\int_0^1\E\left[(A_{t_{n}, t_{n+1}}(X_{s}, X'_{s}))^2\right] \dee s = 2\SKL(\pi^\phi_{t_{n}},\pi^\phi_{t_{n+1}}).
\end{align*}
Therefore we have
\begin{align}\label{eq:bound}
\frac{2\Lambda^\phi(\schedule)^2}{N}\leq \sum_{n=0}^{N-1}\SKL(\pi^\phi_{t_{n}},\pi^\phi_{t_{n+1}}).
\end{align}

The slack in the inequality in Equation \eqref{eq:bound} could potentially depend on $\phi$ even in the large $N$ regime. Therefore, during optimization, we recommend monitoring the value of the original objective function (Equation~(\ref{eq:tauopt})) to ensure that the optimization of the surrogate SKL objective indeed improves it, and hence the round trip rate performance of PT via Equation (\ref{eq:nrpt_rtr}). In the experiments we display the values of both objective functions.

\section{Spline annealing path family}\label{sec:expfamspline}
In this section, we develop a family of annealing paths---the
\emph{spline annealing path family}---that offers a practical and flexible
improvement upon the traditional linear paths considered in past work.  We
first define a general family of annealing paths based on the exponential
family, and then provide the specific details of the spline family with
a discussion of its properties. Empirical results in Section~\ref{sec:experiments} 
demonstrate that the spline annealing path family resolves the problematic
Gaussian annealing example in Figure \ref{fig:gaussian_paths}.

\subsection{Exponential annealing path family}\label{sec:annealed_exponential_family}
We begin with the practical desiderata for an annealing path family given a
fixed reference $\pi_0$ and target $\pi_1$ distribution.\footnote{A natural
extension of this discussion would include parametrized variational reference
distribution families. For simplicity we restrict to a fixed reference.} First,
the traditional linear path $\pi_t\propto \pi_0^{1-t}\pi_1^t$ should be a
member of the family, so that one can achieve at least the round trip rate
provided by that path.  Second, the family should be broadly applicable and not
depend on particular details of either $\pi_0$ or $\pi_1$.  Finally, using the
Gaussian example from Figure \ref{fig:gaussian_paths} and Proposition
\ref{prop:badlinear} as insight, the family should enable the path to smoothly
vary from $\pi_0$ to $\pi_1$ while inflating / deflating the variance as
necessary.

These desiderata motivate the design of the \emph{exponential annealing path family},
in which each annealing path takes the form 
\[
\pi_t\propto \pi_0^{\eta_0(t)}\pi_1^{\eta_1(t)}=\exp(\eta(t)^TW(x)),
\]
for some function $\eta(t)=(\eta_0(t),\eta_1(t))$ and reference/target log
densities $W(x)=(W_0(x),W_1(x))$. Intuitively, $\eta_0(t)$ and $\eta_1(t)$
represent the level of annealing for the reference and target respectively
along the path. 
Proposition \ref{prop:expfampath} shows that a broad collection of functions
$\eta$ indeed construct a valid annealing path family including the linear
path. 
\begin{proposition}\label{prop:expfampath}
Let $\Omega\subseteq\R^2$ be the set
\begin{align*}
\Omega \!=\! \left\{\xi\in\R^2 \! :\!\! \int\exp(\xi^TW(x))\dee x < \infty\right\}.
\end{align*}
Suppose $(0,1)\in\Omega$ and 
$\mathcal{A}$ is a set of piecewise twice continuously differentiable 
functions $\eta : [0, 1] \to \Omega$ such that $\eta(1) = (0, 1)$.
Then 
\[
\left\{\pi_t(x) \propto \exp(\eta(t)^TW(x)) : \eta \in \mathcal{A}\right\}
\]
is an annealing path family for target distribution $\pi_1$.
If additionally $(1,0)\in\Omega$, then the linear path $\eta(t) = (1-t, t)$ may be included
in $\mathcal{A}$.
Finally, if for every $\eta \in \mathcal{A}$ there exists $M > 0$
such that $\sup_t \max\{\|\eta'(t)\|_2, \|\eta''(t)\|_2\} \leq M$
and Equation (\ref{eq:mgf}) holds with $V_1 = V_2 = M\|W\|_2$, 
then each path in the family satisfies the conditions of Theorem~\ref{thm:general_GCB}. 
\end{proposition}

\subsection{Spline annealing path family}
Proposition \ref{prop:expfampath} reduces the problem of designing a general
family of paths of probability distributions to the much simpler task of
designing paths in $\R^2$. We argue that a good choice can be constructed using
linear spline paths connecting $K$ knots 
$\phi = (\phi_0, \dots, \phi_{K}) \in (\R^2)^{K+1}$, i.e., for all 
$k\in\{1, \dots, K\}$ and $t\in [\frac{k-1}{K},\frac{k}{K}]$,
\[
\eta^\phi(t) \mapsto (k-Kt)\phi_{k-1} + (Kt-k+1)\phi_k.
\]
Let $\Omega$ be defined as in Proposition \ref{prop:expfampath}.
The $K$-knot \emph{spline annealing path family} is defined as the set of $K$-knot linear spline paths such that
\[
\phi_0 = (1, 0), \quad \phi_K = (0, 1), \quad \text{and}\quad \forall k,\,\, \phi_k \in \Omega.
\]
\paragraph{Validity:} Since $\Omega$ is a convex set per the proof of Proposition \ref{prop:expfampath}, we 
are guaranteed that $\eta^{\phi}([0,1])\subseteq \Omega$, and so this collection of annealing paths is a 
subset of the exponential annealing family and hence forms a valid 
annealing path family targeting $\pi_1$.

\paragraph{Convexity:} Furthermore, convexity of $\Omega$ implies that tuning the knots
$\phi \in \Omega^{K+1}$ involves optimization within a convex constraint set. In practice, we enforce also
that the knots are monotone in each component---i.e., the first component monotonically decreases, $1 = \phi_{0,0} \ge \phi_{1,0} \ge \dots \ge \phi_{K,0} = 0$  and the
second increases, $0 = \phi_{0,1} \le \phi_{1,1} \le \dots \le \phi_{K,1} = 1$ ---such that the path of distributions always moves from the reference to the target.
Because monotonicity constraint sets are linear and hence convex, the overall monotonicity-constrained optimization 
problem has a convex domain. 

\paragraph{Flexibility:}
 Assuming the family is nonempty, it trivially contains the linear path.
 Further, given a large enough number of knots $K$, the spline annealing family well-approximates
subsets of the exponential annealing family for fixed $M>0$. In particular,
\begin{equation}
\sup_{\eta\in\mathcal{A}_M} \inf_{\phi\in\Omega^{K+1}} \|\eta^\phi-\eta\|_\infty \leq \frac{M}{4K^2}.
\end{equation}
Figure \ref{fig:spline} provides an illustration of the behaviour of optimized spline paths
for a Gaussian reference and target. The path takes a convex curved shape; starting at the bottom right 
point of the figure (reference), this path corresponds to increasing the variance of the reference,
shifting the mean from reference to target, and finally decreasing the variance to match the target.
With more knots, this process happens more smoothly. 

Appendix \ref{sec:gradient} provides an explicit derivation of the stochastic gradient estimates
we use to optimize the knots of the spline annealing path family.
It is also worth making two practical notes. 
First, to enforce the monotonicity constraint, we developed an alternative to the usual projection approach, as projecting into the set of monotone splines can cause several knots to become superposed. 
Instead, we maintain monotonicity
as follows: after each stochastic gradient step, we identify a monotone subsequence of knots containing the endpoints,
remove the nonmonotone jumps, and linearly interpolate between the knot subsequence with an even spacing. 
Second, we take gradient steps in a log-transformed space so that knot components are always strictly positive.

\begin{figure}[t!]
\vskip 0.2in
\begin{center}
\centerline{\includegraphics[width=.9\columnwidth]{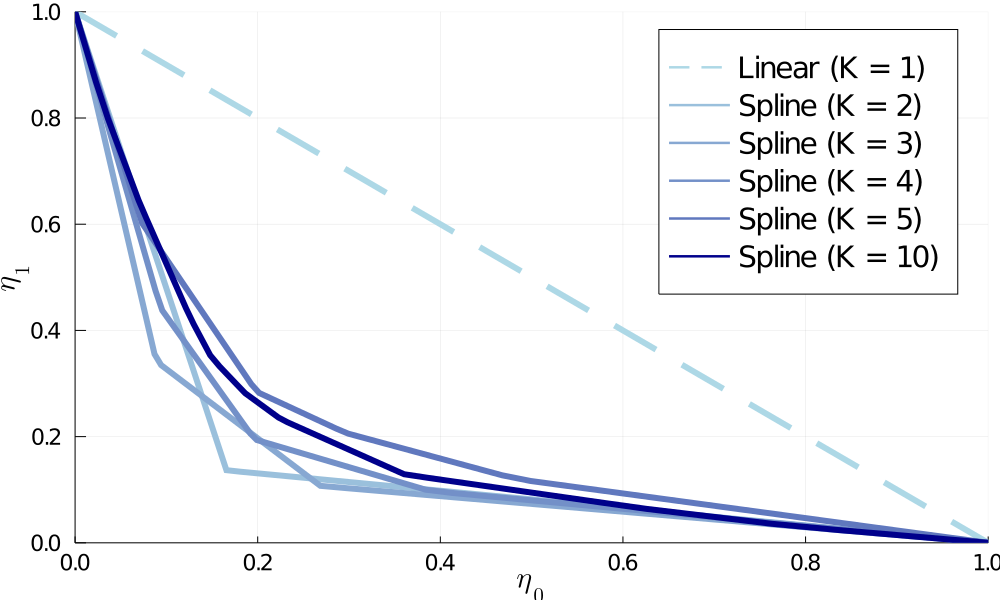}}
\caption{The spline path for $K=1,2,3,4,5,10$ knots for the family 
generated by $\pi_0 = N(-1,0.5 )$ and $\pi_1=N(1,0.5)$. }
\label{fig:spline}
\end{center}
\vskip -0.2in
\end{figure}

\begin{figure*}[t!]
\vskip 0.2in
\begin{center}
\centerline{
\includegraphics[width=0.68\columnwidth]{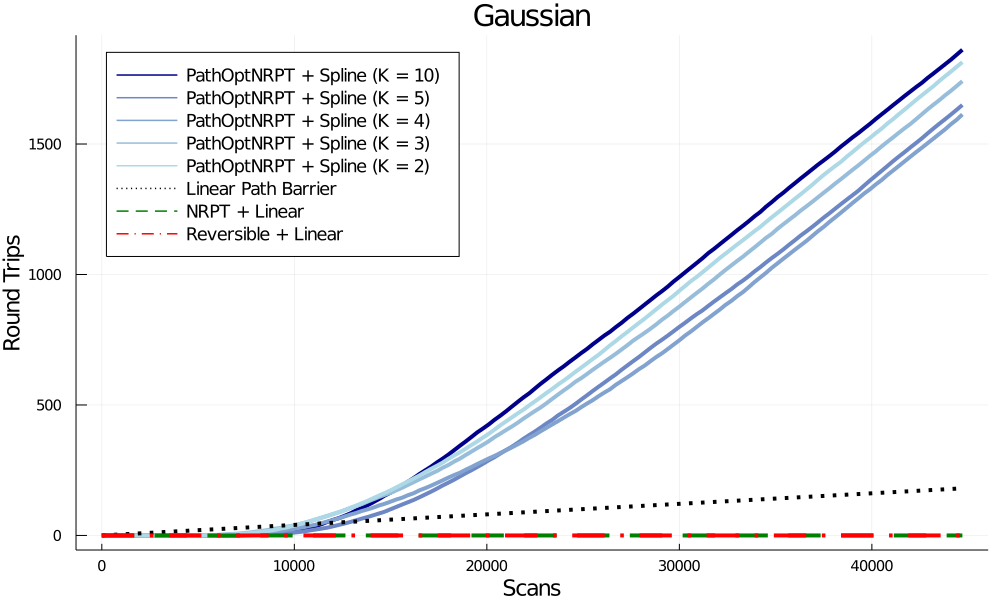}
\includegraphics[width=0.68\columnwidth]{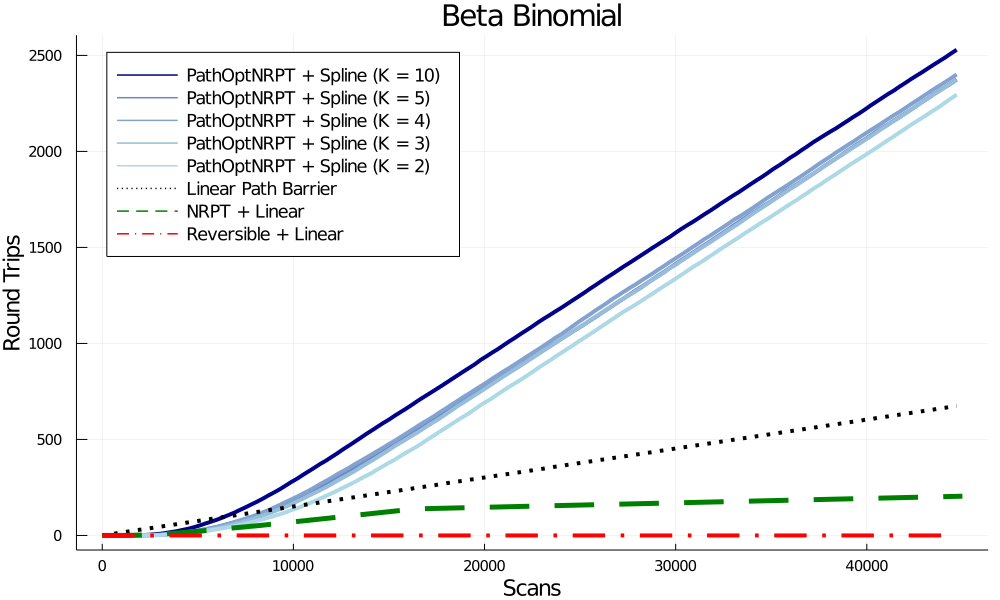}
\includegraphics[width=0.68\columnwidth]{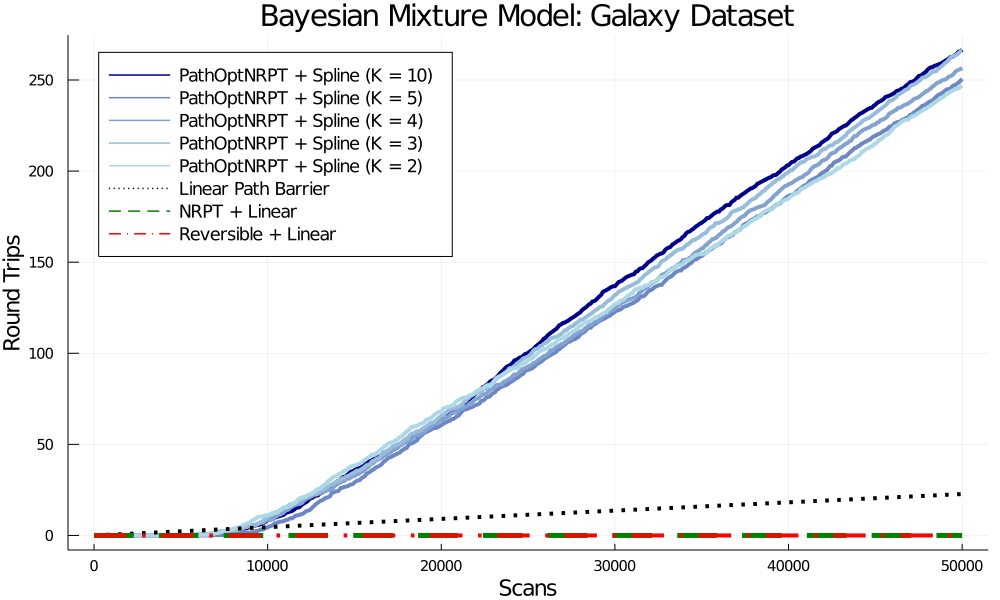}
}
\centerline{
\includegraphics[width=0.68\columnwidth]{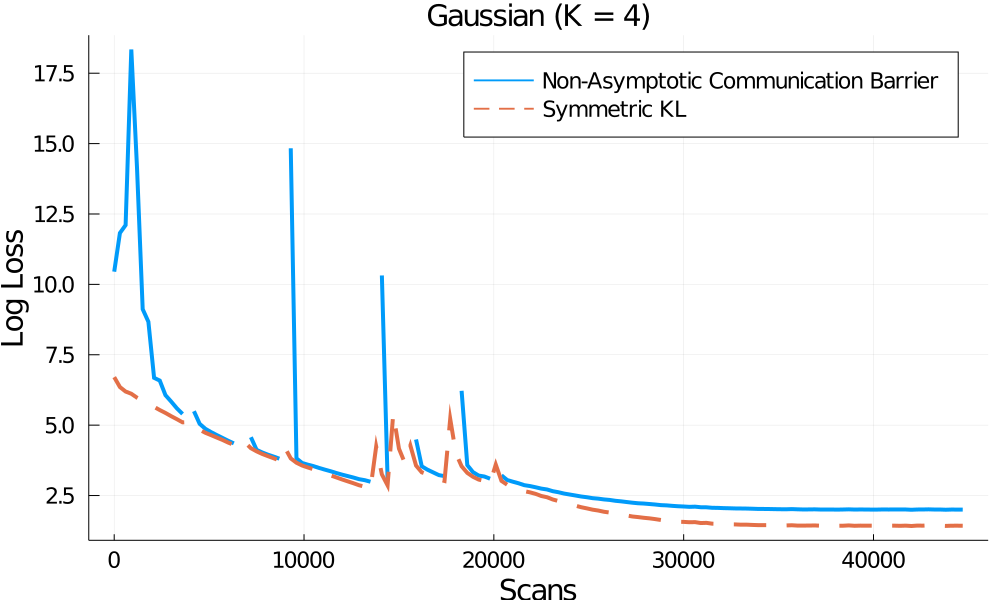}
\includegraphics[width=0.68\columnwidth]{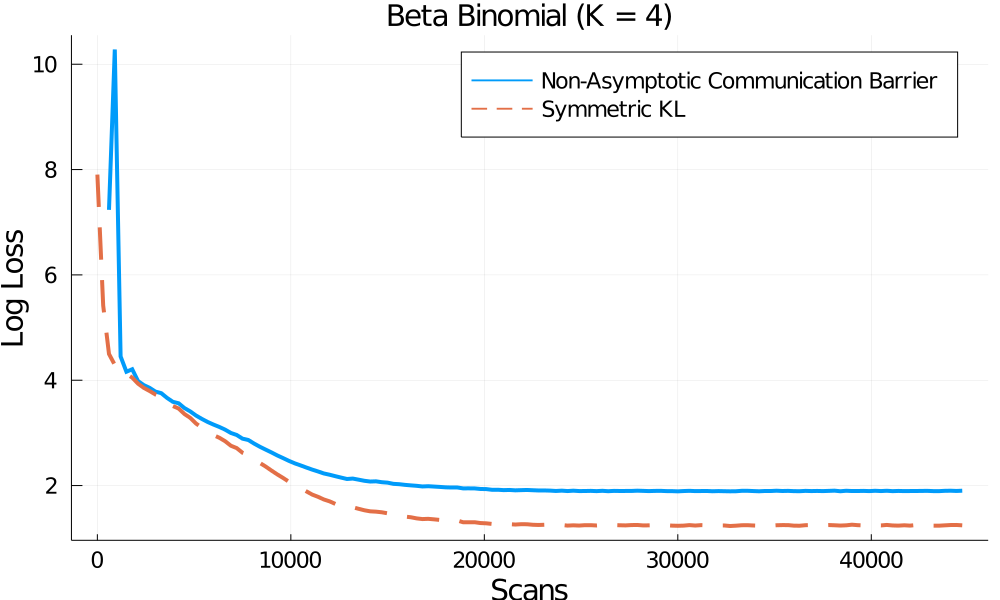}
\includegraphics[width=0.68\columnwidth]{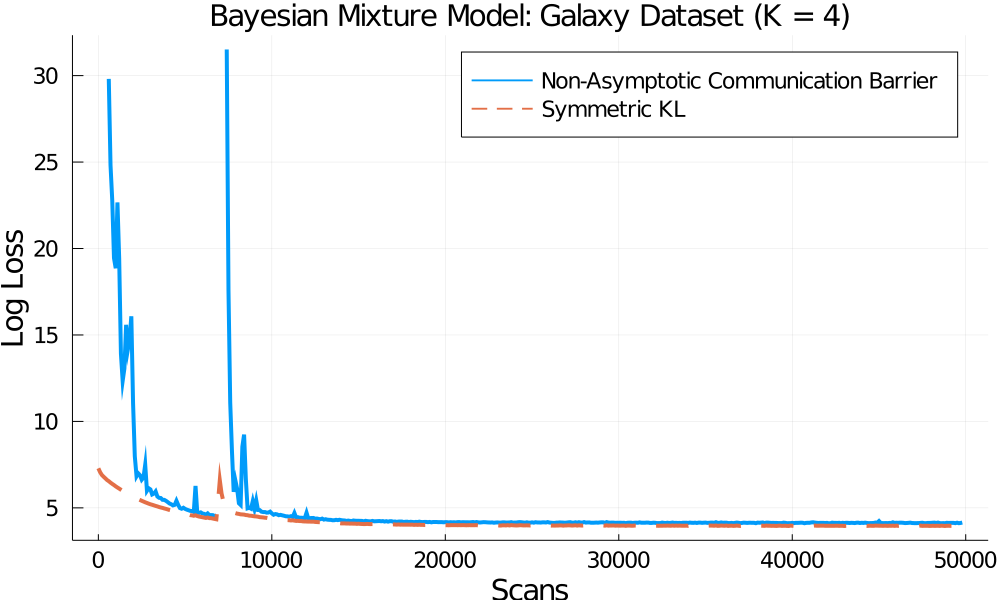}
}
\caption{\textbf{Top:} Cumulative round trips averaged over 10 runs for the spline path with $K=2,3,4,5,10$ (solid blue), NRPT using a linear path (dashed green), and reversible PT with linear path (dash/dot red). The slope of the lines represent the round trip rate. We observe large gains going from linear to non-linear paths ($K > 1$). For all values of $K>1$, the optimized spline path substantially improves on the theoretical upper bound on round trip rate possible using linear path (dotted black). \textbf{Bottom:} Non-asymptotic communication barrier from Equation \ref{eq:nonasympcb} (solid blue) and Symmetric KL (dash orange) as a function of iteration for one run of PathOptNRPT + Spline ($K=4$ knots).}
\label{fig:round_trips}
\end{center}
\vskip -0.2in
\end{figure*}

\section{Experiments}\label{sec:experiments}

In this section, we study the empirical performance of non-reversible PT based on the 
spline annealing path family ($K \in \{2, 3, 4, 5, 10\}$) from Section~\ref{sec:expfamspline}, with knots
and schedule optimized using the tuning method from Section~\ref{sec:general_annealing}. 
We compare this method to two PT methods based on standard linear paths:
non-reversible PT with adaptive schedule (``NRPT+Linear'') \cite{syed_non_reversible_2019},
and reversible PT (``Reversible+Linear'') \cite{atchade_towards_2011}.
Code for the experiments is available at \url{https://github.com/vittrom/PT-pathoptim}.

We use the terminology ``scan'' to denote one iteration
of the for loop in Algorithm~2. The computational cost of a scan is
comparable for all the methods, since the bottleneck is the local
exploration step shared by all methods.
These experiments demonstrate two primary conclusions: (1) tuned nonlinear paths provide
a substantially higher round trip rate compared to linear paths for the examples considered; and (2)
the symmetric KL sum objective (Equation~\ref{eq:bound}) is a good proxy for the 
round trip rate as a tuning objective.


We run the following benchmark problems; see the supplement for details. 
\textbf{Gaussian:} a synthetic setup in which
the reference distribution is $\pi_0=N(-1,0.01^2)$ and the target is
$\pi_1=N(1,0.01^2)$. For this example we used $N=50$ parallel chains and fixed
the computational budget to 45000 samples. For Algorithm \ref{alg:PT_tuning}, the computational
budget was divided equally over 150 scans, meaning 300 samples were used for
every gradient update. ``Reversible+Linear'' performed 45000 local exploration
steps with a communication step after every iteration while for ``NRPT+Linear''
the computational budget was used to adapt the schedule. The gradient updates
were performed using Adagrad \citep{duchi2011adaptive} with learning rate equal to 0.2.
\textbf{Beta-binomial model:} a conjugate Bayesian model
with prior $\pi_0(p) = \mathrm{Beta}(180,840)$ and likelihood
$L(x|p)=\mathrm{Binomial}(x|n,p)$. We simulated data $x_1,\dots,x_{2000}\sim
\mathrm{Binomial}(100,0.7)$ resulting in a posterior $\pi_1(p)=
\mathrm{Beta}(140180,60840)$. The prior and posterior are heavily  concentrated
at 0.2 and 0.7 respectively. We used the same settings as for the Gaussian
example.
\textbf{Galaxy data:} A Bayesian Gaussian mixture model applied to the galaxy dataset of \citep{roeder1990density}. We used six mixture components with mixture proportions $w_0, \ldots, w_5$, mixture component densities $N(\mu_i, 1)$ for mean parameters $\mu_0, \ldots, \mu_5$, and a cluster label categorical variable for each data point. We placed a $\mathrm{Dir}(\boldsymbol{1}_6)$ prior on the
proportions, where $\boldsymbol{1}_6 = (1,1,1,1,1,1)$ and a $N(150, 1)$ prior on each of the mean parameters. We did not marginalize the cluster indicators, creating a multi-modal posterior inference problem over 94 latent variables. In this experiment we used $N=35$ chains and fixed the computational
budget to 50000 samples, divided into 500 scans using 100 samples each. We
optimized the path using Adagrad with a learning rate of 0.3.
Exploring the full posterior distribution is challenging in this context due to a combination of misspecification of the prior and label switching. Label switching refers to the invariance of the likelihood under relabelling of the cluster labels. In Bayesian problems label switching can lead to increased difficulty of the sampling problem as it generates symmetric multi-modal posterior distributions.
\textbf{High dimensional Gaussian:} a similar setup to the one-dimensional Gaussian experiment where the number of dimensions ranges from $d=1$ to $d=256$. The reference distribution is  $\pi_0 = N(-\boldsymbol{1}_d, (0.1^2) I_d)$ and the target is $\pi_1 = N(\boldsymbol{1}_d, (0.1^2)I_d)$ where the subscript $d$ indicates the dimensionality of the problem. The number of chains $N$ is set to increase  with dimension at the rate $N = \lceil 15\sqrt{d} \rceil$. We fixed the number of spline knots $K$ to 4 and set the computational budget to 50000 samples divided into 500 scans with 100 samples per gradient update. The gradient updates
were performed using Adagrad with learning rate equal to 0.2. For all the experiments we performed one local exploration step before each
communication step.

\begin{figure}[t!]
\vskip 0.2in
\begin{center}
\centerline{\includegraphics[width=\columnwidth]{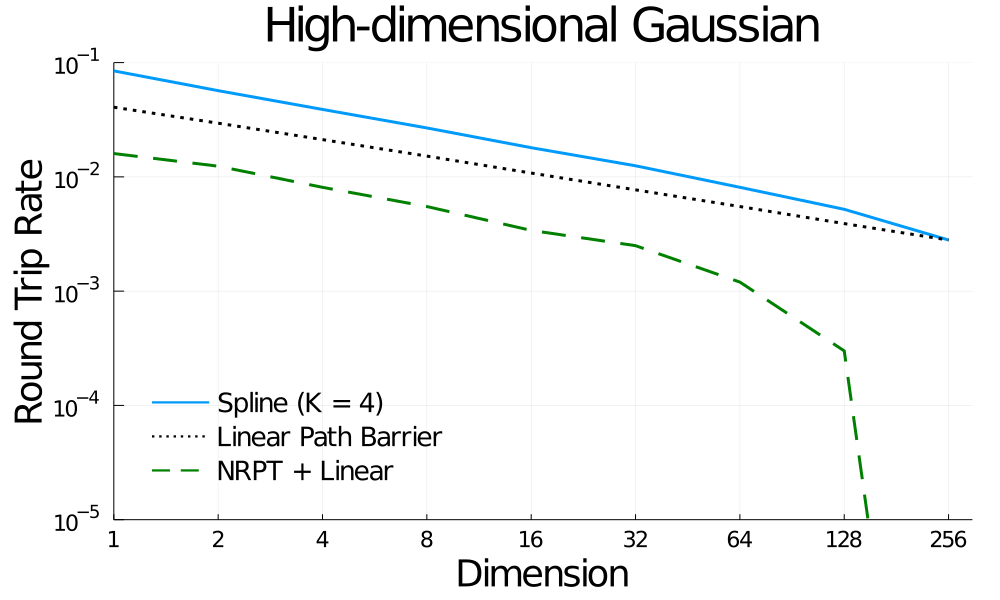}}
\caption{Round trips rate averaged over 5 runs for the spline path with $K=4$ (solid blue), NRPT using a linear path (dashed green) and theoretical upper bound on round trip rate possible using linear path (dotted black) as a function of the dimensionality of the target distribution.}
\label{fig:gaussian_dimensionality}
\end{center}
\vskip -0.2in
\end{figure}

The results of these experiments are shown in Figures~\ref{fig:round_trips} and \ref{fig:gaussian_dimensionality}.
Examining the top row of Figure~\ref{fig:round_trips}---which shows the number of round trips as a function
of the number of scans---one can see that PT using the spline annealing family
outperforms PT using the linear annealing path across all numbers of knots tested.
Moreover, the slope of these curves demonstrates that PT with the spline annealing family 
exceeds the theoretical upper bound of round trip rate for the linear annealing path (from
Equation~(\ref{eq:nonlineartauconv})). The largest gain is obtained from going
from $K=1$ (linear) to $K=2$. For all the examples, 
increasing the number of knots to more than $K>2$ leads to marginal
improvements.  In the case of the Gaussian example, note that since the global
communication barrier $\Lambda$ for the linear path is much larger than $N$,
algorithms based on linear paths incurred rejection rates of nearly one for
most chains, resulting in no round trips. 

The bottom row of Figure~\ref{fig:round_trips} shows the value of the surrogate SKL
objective and non-asymptotic communication barrier from Equation~(\ref{eq:nonasympcb}).
In particular, these figures demonstrate that the SKL provides a surrogate objective
that is a reasonable proxy for the non-asymptotic communication barrier, but does
not exhibit as large estimation variance in early iterations when there are pairs
of chains with rejection rates close to one.

Figure~\ref{fig:gaussian_dimensionality} shows the round trip rate as a function of the dimensionality of the problem for  Gaussian target distributions. As the dimensionality increases, the sampling problem becomes fundamentally more difficult, explaining the decay in performance of both NRPT with a linear path and an optimized path.	
Both sampling methods are provided with a fixed computational budget in all runs. 
Due to this fixed budget, NRPT is unable for $d \ge 64$ to approach the optimal schedule in the alloted time. 
This leads to an increasing gap between the round trip rates of NRPT with linear path and the spline path for $d \ge 64$.

\section{Discussion}
In this work, we identified the use of linear paths of distributions
as a major bottleneck in the performance of parallel tempering algorithms.
To address this limitation, we have provided a theory of parallel tempering on nonlinear
paths, a methodology to tune parametrized paths, and finally
 a practical, flexible family of paths based on linear splines.
Future work in this line of research includes extensions to estimate
normalization constants, as well as the development of techniques 
and theory surrounding the use of variational reference distributions.

\bibliography{references}

\begin{thebibliography}{26}
\providecommand{\natexlab}[1]{#1}
\providecommand{\url}[1]{\texttt{#1}}
\expandafter\ifx\csname urlstyle\endcsname\relax
  \providecommand{\doi}[1]{doi: #1}\else
  \providecommand{\doi}{doi: \begingroup \urlstyle{rm}\Url}\fi

\bibitem[Andrieu \& Moulines(2006)Andrieu and
  Moulines]{andrieu_ergodicity_2006}
Andrieu, C. and Moulines, E.
\newblock On the ergodicity properties of some adaptive {MCMC} algorithms.
\newblock \emph{Annals of Applied Probability}, 16\penalty0 (3):\penalty0
  1462--1505, 2006.

\bibitem[Atchad\'e et~al.(2011)Atchad\'e, Roberts, and
  Rosenthal]{atchade_towards_2011}
Atchad\'e, Y.~F., Roberts, G.~O., and Rosenthal, J.~S.
\newblock Towards optimal scaling of {M}etropolis-coupled {Markov} chain
  {Monte} {Carlo}.
\newblock \emph{Statistics and Computing}, 21\penalty0 (4):\penalty0 555--568,
  2011.

\bibitem[Ballnus et~al.(2017)Ballnus, Hug, Hatz, G{\"o}rlitz, Hasenauer, and
  Theis]{ballnus2017comprehensive}
Ballnus, B., Hug, S., Hatz, K., G{\"o}rlitz, L., Hasenauer, J., and Theis,
  F.~J.
\newblock Comprehensive benchmarking of {M}arkov chain {M}onte {C}arlo methods
  for dynamical systems.
\newblock \emph{BMC Systems Biology}, 11\penalty0 (1):\penalty0 63, 2017.

\bibitem[Brekelmans et~al.(2020)Brekelmans, Masrani, Bui, Wood, Galstyan,
  Steeg, and Nielsen]{brekelmans2020annealed}
Brekelmans, R., Masrani, V., Bui, T., Wood, F., Galstyan, A., Steeg, G.~V., and
  Nielsen, F.
\newblock Annealed importance sampling with q-paths.
\newblock \emph{arXiv:2012.07823}, 2020.

\bibitem[Costa et~al.(2015)Costa, Santos, and Strapasson]{costa2015fisher}
Costa, S.~I., Santos, S.~A., and Strapasson, J.~E.
\newblock Fisher information distance: A geometrical reading.
\newblock \emph{Discrete Applied Mathematics}, 197:\penalty0 59--69, 2015.

\bibitem[Dabak \& Johnson(2002)Dabak and Johnson]{dabak2002relations}
Dabak, A.~G. and Johnson, D.~H.
\newblock Relations between {K}ullback-{L}eibler distance and {F}isher
  information.
\newblock \emph{Journal of The Iranian Statistical Society}, 5:\penalty0
  25--37, 2002.

\bibitem[Desjardins et~al.(2014)Desjardins, Luo, Courville, and
  Bengio]{desjardins2014deep}
Desjardins, G., Luo, H., Courville, A., and Bengio, Y.
\newblock Deep tempering.
\newblock \emph{arXiv:1410.0123}, 2014.

\bibitem[Duchi et~al.(2011)Duchi, Hazan, and Singer]{duchi2011adaptive}
Duchi, J., Hazan, E., and Singer, Y.
\newblock Adaptive subgradient methods for online learning and stochastic
  optimization.
\newblock \emph{Journal of machine learning research}, 12\penalty0 (7), 2011.

\bibitem[Gelman \& Meng(1998)Gelman and Meng]{gelman1998simulating}
Gelman, A. and Meng, X.-L.
\newblock Simulating normalizing constants: From importance sampling to bridge
  sampling to path sampling.
\newblock \emph{Statistical science}, pp.\  163--185, 1998.

\bibitem[Geyer(1991)]{geyer1991markov}
Geyer, C.~J.
\newblock Markov chain {M}onte {C}arlo maximum likelihood.
\newblock \emph{Interface Proceedings}, 1991.

\bibitem[Grosse et~al.(2013)Grosse, Maddison, and
  Salakhutdinov]{grosse2013annealing}
Grosse, R.~B., Maddison, C.~J., and Salakhutdinov, R.~R.
\newblock Annealing between distributions by averaging moments.
\newblock In \emph{Advances in Neural Information Processing Systems}, 2013.

\bibitem[Kamberaj(2020)]{kamberaj2020molecular}
Kamberaj, H.
\newblock \emph{Molecular Dynamics Simulations in Statistical Physics: Theory
  and Applications}.
\newblock Springer, 2020.

\bibitem[Katzgraber et~al.(2006)Katzgraber, Trebst, Huse, and
  Troyer]{katzgraber2006feedback}
Katzgraber, H.~G., Trebst, S., Huse, D.~A., and Troyer, M.
\newblock Feedback-optimized parallel tempering {M}onte {C}arlo.
\newblock \emph{Journal of Statistical Mechanics: Theory and Experiment},
  2006\penalty0 (03):\penalty0 P03018, 2006.

\bibitem[Lingenheil et~al.(2009)Lingenheil, Denschlag, Mathias, and
  Tavan]{lingenheil2009efficiency}
Lingenheil, M., Denschlag, R., Mathias, G., and Tavan, P.
\newblock Efficiency of exchange schemes in replica exchange.
\newblock \emph{Chemical Physics Letters}, 478\penalty0 (1-3):\penalty0 80--84,
  2009.

\bibitem[M{\"u}ller \& Bouckaert(2020)M{\"u}ller and
  Bouckaert]{mueller2020adaptive}
M{\"u}ller, N.~F. and Bouckaert, R.~R.
\newblock Adaptive parallel tempering for {BEAST} 2.
\newblock \emph{bioRxiv}, 2020.
\newblock \doi{10.1101/603514}.

\bibitem[Okabe et~al.(2001)Okabe, Kawata, Okamoto, and
  Mikami]{okabe2001replica}
Okabe, T., Kawata, M., Okamoto, Y., and Mikami, M.
\newblock Replica-exchange {M}onte {C}arlo method for the isobaric--isothermal
  ensemble.
\newblock \emph{Chemical Physics Letters}, 335\penalty0 (5-6):\penalty0
  435--439, 2001.

\bibitem[Predescu et~al.(2004)Predescu, Predescu, and
  Ciobanu]{predescu2004incomplete}
Predescu, C., Predescu, M., and Ciobanu, C.~V.
\newblock The incomplete beta function law for parallel tempering sampling of
  classical canonical systems.
\newblock \emph{The Journal of Chemical Physics}, 120\penalty0 (9):\penalty0
  4119--4128, 2004.

\bibitem[Rainforth et~al.(2018)Rainforth, Kosiorek, Le, Maddison, Igl, Wood,
  and Teh]{RainforthKLMIWT18}
Rainforth, T., Kosiorek, A.~R., Le, T.~A., Maddison, C.~J., Igl, M., Wood, F.,
  and Teh, Y.~W.
\newblock Tighter variational bounds are not necessarily better.
\newblock In \emph{International Conference on Machine Learning}, 2018.

\bibitem[Rischard et~al.(2018)Rischard, Jacob, and
  Pillai]{rischard2018unbiased}
Rischard, M., Jacob, P.~E., and Pillai, N.
\newblock Unbiased estimation of log normalizing constants with applications to
  {B}ayesian cross-validation.
\newblock \emph{arXiv:1810.01382}, 2018.

\bibitem[Roeder(1990)]{roeder1990density}
Roeder, K.
\newblock Density estimation with confidence sets exemplified by superclusters
  and voids in the galaxies.
\newblock \emph{Journal of the American Statistical Association}, 85\penalty0
  (411):\penalty0 617--624, 1990.

\bibitem[Sakai \& Hukushima(2016)Sakai and Hukushima]{sakai_irreversible_2016}
Sakai, Y. and Hukushima, K.
\newblock Irreversible simulated tempering.
\newblock \emph{Journal of the Physical Society of Japan}, 85\penalty0
  (10):\penalty0 104002, 2016.

\bibitem[Syed et~al.(2019)Syed, Bouchard-C\^{o}t\'{e}, Deligiannidis, and
  Doucet]{syed_non_reversible_2019}
Syed, S., Bouchard-C\^{o}t\'{e}, A., Deligiannidis, G., and Doucet, A.
\newblock Non-reversible parallel tempering: an embarassingly parallel {MCMC}
  scheme.
\newblock \emph{arXiv:1905.02939}, 2019.

\bibitem[Tawn et~al.(2020)Tawn, Roberts, and Rosenthal]{tawn2020weight}
Tawn, N.~G., Roberts, G.~O., and Rosenthal, J.~S.
\newblock Weight-preserving simulated tempering.
\newblock \emph{Statistics and Computing}, 30\penalty0 (1):\penalty0 27--41,
  2020.

\bibitem[Whitfield et~al.(2002)Whitfield, Bu, and
  Straub]{whitfield2002generalized}
Whitfield, T., Bu, L., and Straub, J.
\newblock Generalized parallel sampling.
\newblock \emph{Physica A: Statistical Mechanics and its Applications},
  305\penalty0 (1-2):\penalty0 157--171, 2002.

\bibitem[Woodard et~al.(2009)Woodard, Schmidler, Huber,
  et~al.]{woodard2009rapid}
Woodard, D.~B., Schmidler, S.~C., Huber, M., et~al.
\newblock Conditions for rapid mixing of parallel and simulated tempering on
  multimodal distributions.
\newblock \emph{The Annals of Applied Probability}, 19\penalty0 (2):\penalty0
  617--640, 2009.

\bibitem[Zhou et~al.(2016)Zhou, Johansen, and Aston]{Zhou2016}
Zhou, Y., Johansen, A.~M., and Aston, J.~A.
\newblock Toward automatic model comparison: An adaptive sequential {M}onte
  {C}arlo approach.
\newblock \emph{Journal of Computational and Graphical Statistics}, 25\penalty0
  (3):\penalty0 701--726, 2016.

\end{thebibliography}
\bibliographystyle{icml2021}

\clearpage

\appendix
\section{Proof of Proposition \ref{prop:badlinear}}
Define $\pi_0=N(\mu_0,\sigma^2)$ and $\pi_1=N(\mu_1,\sigma^2)$ with
$W_i(x)\propto-\frac{1}{2\sigma^2}(x-\mu_i)^2$ (throughout we use the proportionality symbol $\propto$
with log-densities to indicate an unspecified  constant, additive with respect to $W_i$, multiplicative with respect to $\pi_t$).
Suppose $\pi_t$ is the linear 
path $\pi_t(x)\propto \exp(W_t)$ where $W_t=(1-t)W_0+tW_1$. Note that
as a function of $x$,
\begin{align*}
W_t(x)
&\propto-\frac{1-t}{2\sigma^2}(x-\mu_0)^2-\frac{t}{2\sigma^2}(x-\mu_1)^2\\
& \propto-\frac{1}{2\sigma^2}(x-\mu_t)^2, \quad \mu_t=(1-t)\mu_0+t\mu_1,
\end{align*}
and thus $\pi_t=N(\mu_t,\sigma^2)$.
Taking a derivative of $W_t$, we find that 
\[
\frac{\dee W_t}{\dee t} = \frac{(\mu_1-\mu_0)\left(x-\frac{\mu_0+\mu_1}{2}\right)}{\sigma^2}.
\]
We will now compute $\lambda(t)$. If $X_t,X'_t\sim \pi_t$, then
\begin{align*}
\lambda(t)
&=\frac{1}{2}\E\left[\left|\frac{dW}{dt}(X_t)-\frac{dW}{dt}(X'_t)\right|\right]\\
&=\frac{|\mu_1-\mu_0|}{2\sigma}\E\left[\left|\frac{X_t-\frac{\mu_0+\mu_1}{2}}{\sigma}-\frac{X_t'-\frac{\mu_0+\mu_1}{2}}{\sigma}\right|\right]\\
&=\frac{|\mu_1-\mu_0|}{2\sigma}\E\left[\left|\frac{X_t-\mu_t}{\sigma}-\frac{X_t'-\mu_t}{\sigma}\right|\right]\\
&=\frac{|\mu_1-\mu_0|}{2\sigma}\E\left[\left|Z-Z'\right|\right],
\end{align*}
where $Z,Z'\sim N(0,1)$. Thus $Z-Z'\sim N(0,2)$, and $|Z-Z'|$ has a folded
normal distribution with expectation $2/\sqrt{\pi}$. This implies
$\lambda(t)=z/\sqrt{\pi}$ where $z=|\mu_1-\mu_0|/\sigma$ and
$\Lambda=\int_0^1\lambda(t)\dee t=z/\sqrt{\pi}$. By Theorem
\ref{thm:general_GCB}, the asymptotic round trip rate
$\tau_\infty^{\mathrm{linear}}$ for the linear path satisfies,
\begin{align*}
\tau_\infty^{\mathrm{linear}}=\frac{1}{2+2\Lambda}=\frac{1}{2+2z/\sqrt{\pi}}=\Theta\left(\frac{1}{z}\right).
\end{align*}

We will now establish an upper bound for the communication barrier $\Lambda$ for a general path $\pi_t$. If $X_t,X_t'\sim\pi_t$, then Theorem \ref{thm:general_GCB} and Jensen's inequality imply the following:
\begin{align*}
\Lambda
&=\int_0^1\frac{1}{2}\E\left[\sqrt{\left(\frac{dW}{dt}(X_t)-\frac{dW}{dt}(X'_t)\right)^2}\right]\dee t\\
&\leq\int_0^1\frac{1}{2}\sqrt{\E\left[\left(\frac{dW}{dt}(X_t)-\frac{dW}{dt}(X'_t)\right)^2\right]}\dee t\\
&=\frac{1}{\sqrt{2}}\int_0^1\sqrt{\Var_{\pi_t} \left[\frac{dW_t}{dt}\right]}\dee t\\
&=\frac{1}{\sqrt{2}}\Lambda_F,
\end{align*}
where $\Lambda_F$ is the length of the the path $\pi_t$ with the Fisher
information metric. The geodesic path of Gaussians between $\pi_0$ and $\pi_1$ 
 that minimizes $\Lambda_F$ satisfies \citep[Eq.~11, Sec.~2]{costa2015fisher}
\begin{equation}
\Lambda_F =\sqrt{2}\log\left(1+\frac{z^2}{4} +\frac{z}{4}\sqrt{8+z^2}\right).
\end{equation} 
Again, by Theorem \ref{thm:general_GCB}, the asymptotic round trip rate
$\tau_\infty^{\mathrm{geodesic}}$ for the geodesic path satisfies
\begin{align*}
\tau_\infty^{\mathrm{geodesic}}=\frac{1}{2+2\Lambda}\geq \frac{1}{2+2\Lambda_F}=\Theta\left(\frac{1}{\log z}\right).
\end{align*}

\section{Proof of Lemma \ref{lemma:nonlinear_gcb_err}}
\begin{definition}
Given a path $\pi_t$ and measurable function $f$, we denote $\|f\|_\pi =\sup_t\E_{\pi_t}[f]$.
\end{definition}
Following the computation in \citet[Equation (6)]{predescu2004incomplete}, we have
\begin{align}\label{eq:rejection_formula_linear}
r(t,t')=1-\frac{\E[\exp(-\frac{1}{2}|A_{t,t'}(\tilde{X}_{1/2},\tilde{X}'_{1/2})|]}{\E[\exp(-\frac{1}{2}A_{t,t'}(\tilde{X}_{1/2},\tilde{X}'_{1/2}))]},
\end{align}
where $\tilde{X}_s,\tilde{X}'_s\sim \tilde{\pi}_s= \frac{1}{\tilde{Z}(s)}\exp((1-s)W_t+sW_{t'})$ and 
\[
A_{t,t'}(x,x')=(W_{t'}(x)-W_t(x))-(W_{t'}(x')-W_t(x')).
\]
In particular, the path of distributions $\tilde{\pi}_s$ for $s\in[0,1]$ is the linear path between $\pi_t$ and $\pi_{t'}$.

\begin{lemma}\label{lemma:acceptance_estimate}
Suppose (\ref{eq:V1_bound}) and (\ref{eq:V1_integrability}) hold. 
Then for all $k\leq 3$, there is a constant $\tilde{C}_k$ independent of $t,t',\schedule$ such that 
\[
\sup_s\E[|A_{t,t'}(\tilde{X}_s,\tilde{X}'_s)|^k]\leq \tilde{C}_k|t'-t|^k.
\]
where $\tilde{X}_s,\tilde{X}_s'\sim\tilde{\pi}_s$.
\end{lemma}

\begin{proof}
The mean-value theorem and (\ref{eq:V1_bound}) imply that $W_t(x)$ is Lipschitz in $t$,
\begin{equation*}
|W_t(x)-W_{t'}(x)|\leq V_1(x)|t'-t|.
\end{equation*}
The triangle inequality therefore implies
\[
|A_{t,t'}(x,x')|\leq (V_1(x)+V_1(x'))|t'-t|.
\]
By taking expectations and using the fact $|a+b|^k\leq 2^{k-1}(|a|^k+|b|^k)$, we have that
\begin{align*}
\E[|A_{t,t'}(\tilde X_s,\tilde X'_s)|^k]
&\leq 2^k\E_{\tilde{\pi}_s}[V_1^k]|t'-t|^k\\
&\leq 2^k\E_{\tilde{\pi}_s}[V_1^3]|t'-t|^k,
\end{align*}
where in the last line we use the fact that we can assume $V_1 \ge 1$ without loss of generality. 
The result follows by taking the supremum on both sides and noting that $\tilde{C}_k = 2^k\|V_1^3\|_{\tilde{\pi}}$ is finite by  (\ref{eq:V1_integrability}).
\end{proof}

We now begin the proof of Lemma \ref{lemma:nonlinear_gcb_err}.
Define $\tilde{\lambda}(s)=\frac{1}{2}\E[|A_{t,t'}(\tilde{X}_s,\tilde{X}'_s)|]$ for $\tilde{X}_s,\tilde{X}'_s\sim \tilde{\pi}_s$.
Then a third order Taylor expansion of Equation \eqref{eq:rejection_formula_linear} \citep{predescu2004incomplete},
which contains terms of the form $\E[|A_{t,t'}(\tilde{X}_s,\tilde{X}'_s)|^k]$ that can be 
controlled via Lemma \ref{lemma:acceptance_estimate}, yields
\begin{align*}
r(t,t')&=\tilde{\lambda}(1/2)+R(t,t'), \quad |R(t,t')| \leq C'|t-t'|^3,
\end{align*}
for some finite constant $C'$ independent of $t, t'$.
%
By \citet[Prop.~2, Appendix C]{syed_non_reversible_2019} we 
have that in addition, $\tilde\lambda(s)$ is in $C^2([0,1])$,
and thus there is a constant $C''$ independent of $t,t'$ such that
\[
\sup_s\left|\frac{d^2\tilde\lambda}{ds^2}\right|\leq C''|t-t'|^3.
\]
The error bound for the midpoint rule implies,
\begin{align*}
\left|\tilde\lambda(1/2)-\int_0^1\tilde\lambda(s)ds\right|
&\leq \frac{1}{24}\sup_s\left|\frac{d^2\tilde\lambda}{ds^2}\right|\\
&\leq \frac{C''}{24}|t-t'|^3.
\end{align*}
The result follows:  there is a finite constant $C$ independent of $t,t'$ such that
\begin{align*}
\left|r(t,t') - \Lambda(t,t')\right| = \left| r(t,t')-\int_0^1\tilde{\lambda}(s)ds\right| \leq C|t'-t|^3.
\end{align*}

\section{Proof of Theorem \ref{thm:general_GCB}}
We first note that without loss of generality we can place an artificial schedule point $t_n$ at 
each of the finitely many discontinuities in $W_t$ or its first/second derivative.
Thus we assume the $W_t$ is $C^2$ on
each interval $[t_{n-1},t_n]$. Later in the proof it will become clear that the contributions
of these artificial schedule points becomes negligible as $\|\schedule\|\to 0$.

Given a schedule $\schedule$, define the path 
$\tilde\pi_t=\frac{1}{\tilde Z_t} \exp(\tilde W_t)$ 
with log-likelihood $\tilde W_t$ 
satisfying for each segment $t_{n-1}\leq t \leq  t_n$,
\[
\tilde W_t = W_{t_{n-1}}+\frac{\Delta W_n}{\Delta t_n}(t-t_{n-1}), 
\]
where $\Delta W_n= W_{t_n}-W_{t_{n-1}}$ and $\Delta t_n = t_{n}-t_{n-1}$. In
particular, $\tilde{W}_t$ agrees with $W_t$ for $t\in\schedule$, linearly
interpolates between $W_{t_{n-1}}$ and $W_{t_n}$ for $t\in[t_{n-1},t_n]$,
and for all $x$, from Taylor's theorem:
\begin{equation}\label{eq:spline_error}
|\tilde W_t(x)-W_t(x)|\leq \frac{1}{2}\sup_{t\in[t_{n-1},t_n]}\left|\frac{d^2W_t}{dt^2}(x)\right|\Delta t_n^2.
\end{equation}
The following lemma shows that the normalization constant of, and expectations under, $\tilde \pi_t$ 
are comparable to the same for $\pi_t$ with an error bound that depends on $\|\schedule\|$
and converges to 0 as $\|\schedule\|\to 0$.

\begin{lemma}\label{lemma:spline_expectation_error} 
For measurable functions $f$ and $s > 0$, let 
\[
E_t(f, s) = \E_{\pi_t}\left[|f|e^{s^2 V_2}\right],
\]
and define $E_t(s) = E_t(1, s)$ for brevity.
\begin{enumerate}
\item[(a)]  For any schedule $\schedule$, 
\begin{align*}
\left|\frac{\tilde Z_t}{Z_t} - 1\right| \leq
E_t(\|\schedule\|) - 1,
\end{align*}
and if $\|\schedule\|$ is small enough that $E_t(\|\schedule\|) < 2$,
\begin{align*}
\left|\frac{Z_t}{\tilde Z_t} - 1\right| \leq
\frac{E_t(\|\schedule\|)-1}{2-E_t(\|\schedule\|)}.
\end{align*}
\item[(b)] For any schedule $\schedule$ and measureable function $f$,
if $\|\schedule\|$ is small enough that $E_t(\|\schedule\|) < 2$,
\begin{align*}
\left|\E_{\tilde\pi_t}[f] - \E_{\pi_t}[f]\right| &\leq  \frac{E_t(\|\schedule\|)-1}{2-E_t(\|\schedule\|)}E_t(f, \|\schedule\|)\\
& +E_t(f, \|\schedule\|) - E_t(f, 0).
\end{align*}
\end{enumerate}
\end{lemma}
\begin{proof}
\begin{enumerate}
\item[(a)] We rewrite the expression
\begin{align*}
\frac{\tilde{Z}_t}{Z_t}
&=\frac{1}{Z_t}\int_\statespace e^{\tilde{W}_t(x)}\dee x\\
&=\int_\statespace e^{\tilde{W}_t(x)-W_t(x)}\pi_t(x)\dee x\\
&=1 + \int_\statespace \left(e^{\tilde{W}_t(x)-W_t(x)} - 1\right)\pi_t(x)\dee x.
\end{align*}
Thus using the inequality $|e^x-1| \leq e^{|x|}-1$,
\begin{align*}
\left|\frac{\tilde{Z}_t}{Z_t} - 1\right| &\leq \left|\int_\statespace \left(e^{\tilde{W}_t(x)-W_t(x)} - 1\right)\pi_t(x)\dee x\right|\\
&\leq \int_\statespace \left(e^{|\tilde{W}_t(x)-W_t(x)|} - 1\right)\pi_t(x)\dee x\\
&\leq \int_\statespace \left(e^{V_2(x)\|\schedule\|^2} - 1\right)\pi_t(x)\dee x\\
&= \E_{\pi_t}\left[e^{\|\schedule\|^2 V_2}-1\right]\\
&= E_t(\|\schedule\|)-1.
\end{align*}
The bound on $|Z_t/\tilde Z_t - 1|$ arises from straightforward algebraic manipulation of the above bound.

\item[(b)] We begin by rewriting $\E_{\tilde{\pi}_t}[f]$:
\begin{align*}
&\E_{\tilde{\pi}_t}[f] -\E_{\pi_t}[f]\\
&=\frac{1}{\tilde{Z}_t}\int_\statespace f(x)e^{\tilde{W}_t(x)}\dee x-\E_{\pi_t}[f]\\
&=\int_\statespace f(x)\left(\frac{Z_t}{\tilde{Z}_t}e^{\tilde{W}_t(x)-W_t(x)}-1\right)\pi_t(x)\dee x\\
&=\left(\frac{Z_t}{\tilde{Z}_t}-1\right)\int_\statespace f(x)e^{\tilde{W}_t(x)-W_t(x)}\pi_t(x)\dee x \\
&+\int_\statespace f(x)\left(e^{\tilde{W}_t(x)-W_t(x)}-1\right)\pi_t(x)\dee x.
\end{align*}
Therefore again using $|e^x-1| \leq e^{|x|}-1$ and the previous bound,
\begin{align*}
\left|\E_{\tilde{\pi}_t}[f] -\E_{\pi_t}[f]\right|
&\leq \frac{E_t(\|\schedule\|)-1}{2-E_t(\|\schedule\|)}E_t(f, \|\schedule\|)\\
& +E_t(f, \|\schedule\|) - E_t(f, 0).
\end{align*}
\end{enumerate}
\end{proof}

By changing variables via $t = t_{n-1} +s\Delta t_n$ in \eqref{eq:def_GCB_linear}, we can 
rewrite $\Lambda(t_{n-1},t_n)$ as
\[
\Lambda(t_{n-1},t_n)=\int_{t_{n-1}}^{t_n} \frac{1}{2}\E\left[\left|\frac{\Delta W_n}{\Delta t_n}(\tilde{X}_t)-\frac{\Delta W_n}{\Delta t_n}(\tilde{X}'_t)\right|\right]\dee t,
\]
where $\tilde{X}_t,\tilde{X}'_t\sim \tilde \pi_t$.  Note that by construction
for $t \in (t_{n-1},t_n)$ we have $\frac{d\tilde{W}_t}{dt}$ exists and equals
$\frac{\Delta W_n}{\Delta t_n}$. So by summing over $n$ we get,
\begin{align*}
\Lambda(\schedule)
&=\sum_{n=1}^N\Lambda(t_{n-1},t_n)\\
&=\int_0^1 \frac{1}{2}\E\left[\left|\frac{d\tilde W_t}{dt}(\tilde X_t)-\frac{d\tilde W_t}{dt}(\tilde X_t')\right|\right]dt\\
&=\int_0^1\tilde\lambda(t)dt
\end{align*}

If we can show that $\sup_{t}|\tilde{\lambda}(t)-\lambda(t)|$ converges uniformly\footnote{We say $a(\schedule)$ converges uniformly to $a$ if for all $\epsilon>0, \exists\delta>0$ such that $\|\schedule\|<\delta$ implies $|a(\schedule)-a|<\epsilon$. } to 0 as $\|\schedule\|\to 0$ then by dominated convergence theorem $\Lambda(\schedule)$ converges to $\Lambda$ uniformly as $\|\schedule\|\to 0$. The round trip rate then uniformly converges to $(2+2\Lambda)^{-1}$ by Theorem 3 of \cite{syed_non_reversible_2019}.

Adding and subtracting 
$\E\left[\left| \frac{d\tilde W_t}{dt}(X_t) - \frac{d\tilde W_t}{dt}(X'_t)\right|\right]$
within the absolute difference $2|\tilde \lambda(t) - \lambda(t)|$ and using the triangle
inequality, it can be shown that we require bounds on
\[
J_{1,t} = \int \pi_t(x)\pi_t(y){\scriptstyle \left| | \frac{d\tilde W_t}{dt}(x) - \frac{d\tilde W_t}{dt}(y)| - | \frac{d W_t}{dt}(x) - \frac{d W_t}{dt}(y)| \right| }
\]
and
\[
J_{2,t} = \int \left|\pi_t(x)\pi_t(y)-\tilde\pi_t(x)\tilde\pi_t(y)\right|\left| \frac{d\tilde W_t}{dt}(x) - \frac{d\tilde W_t}{dt}(y)\right|. 
\]
For the first term, the mean value theorem implies that there exist $s, s' \in [t_{n-1}, t_n]$ (potentially functions of $x$ and $y$, respectively) such that
\[
J_{1,t} = \int \pi_t(x)\pi_t(y) {\scriptstyle \left| | \frac{dW_s}{dt}(x) - \frac{d W_{s'}}{dt}(y)| - | \frac{d W_t}{dt}(x) - \frac{d W_t}{dt}(y)| \right| }
\]
Split the integral into the set $A$ of $x, y \in \statespace$ where the first term
in the absolute value is larger; the same analysis with the same result
applies in the other case in $A^c$. Here, Taylor's theorem and the triangle inequality yield
\begin{align*}
\left| \frac{dW_s}{dt}(x) - \frac{d W_{s'}}{dt}(y)\right| &\leq
\left| \frac{dW_t}{dt}(x) - \frac{d W_{t}}{dt}(y)\right|\\ & + (V_2(x)+V_2(y))\|\schedule\|.
\end{align*}
Using this and the same procedure for $A^c$, we have that
\begin{align*}
J_{1,t} &\leq \int \pi_t(x)\pi_t(y) (V_2(x)+V_2(y))\|\schedule\|\\
 &=2\E_{\pi_t}\left[V_2\right]\|\schedule\|.
\end{align*}
This converges to 0 as $\|\schedule\|\to 0$.

For the second term $J_{2,t}$, we can again use the mean value theorem to
find $s, s' \in[t_{n-1}, t_n]$ where 
\begin{align*}
J_{2,t} &= \int \left|\pi_t(x)\pi_t(y)-\tilde\pi_t(x)\tilde\pi_t(y)\right|\left| \frac{dW_s}{dt}(x) - \frac{d W_{s'}}{dt}(y)\right|, 
\end{align*}
and therefore via the triangle inequality, symmetry, and the $V_1(x)$ bound on the first path derivative,
\begin{align*}
J_{2,t} &\leq 2\int V_1(x) \left|\pi_t(x)\pi_t(y)-\tilde\pi_t(x)\tilde\pi_t(y)\right|. 
\end{align*}
We then add and subtract $\pi_t(x)\tilde\pi_t(y)$ within the absolute value and use the triangle inequality again to find that
\begin{align*}
J_{2,t} &\leq 2\int \left(V_1(x) + \E_{\pi_t}[V_1]\right) \left|\pi_t(x)-\tilde\pi_t(x)\right|\\ 
 &= 2\int \pi_t(x)\left(V_1(x) + \E_{\pi_t}[V_1]\right) \left|1-\frac{\tilde\pi_t(x)}{\pi_t(x)}\right|. 
\end{align*}
Note that by the triangle inequality and the bound $|e^x-1|\leq e^{|x|}-1$,
\begin{align*}
\left|1-\frac{\tilde\pi_t(x)}{\pi_t(x)}\right| &\leq 
\left|\frac{Z_t}{\tilde Z_t}-1\right|e^{\|\schedule\|^2 V_2(x)} +e^{\|\schedule\|^2V_2(x)} - 1.
\end{align*}
Assume that $\|\schedule\|$ is 
small enough such that $E_t(\|\schedule\|) < 2$,
and let $f = V_1 + \E_{\pi_t}V_1$. Then by 
Lemma \ref{lemma:spline_expectation_error}, 
\begin{align*}
J_{2,t} &\leq 2
\frac{E_t(\|\schedule\|)-1}{2-E_t(\|\schedule\|)} E_t(f, \|\schedule\|)\\ 
&+ E_t(f, \|\schedule\|) - E_t(f, 0).
\end{align*}

By assumption we know that $E_t(f, s)$ is finite for some $s$ small enough.
Therefore as $\|\schedule\|\to 0$, by monotone convergence $E_t(f, \|\schedule\|) \to E_t(f, 0)$,
and in particular $E_t(\|\schedule\|) \to 1$. 
Therefore $J_{1,t} + J_{2,t} \to 0$ as $\|\schedule\|\to 0$ and the proof is complete.

\section{Objective and Gradient}\label{sec:gradient}
Here we derive the gradient used to optimize the surrogate SKL objective in
Equation {\ref{eq:bound}}. First we derive the gradient for the expectation of a general
function in Section {\ref{subsec:gradient_expectation}}. Next, in
Section {\ref{subsec:gradient_exponential_linear}}, we show the result for the specific
case of expectations of linear functions with respect to distributions in the
exponential family. Lastly, we show how the result is related to our SKL
objective in Sections {\ref{subsec:SKL_general}} and {\ref{subsec:SKL_expfamily}}.

\subsection{Derivative of parameter-dependent expectation} \label{subsec:gradient_expectation}
Here we consider the problem of computing
\[
g_\phi(x) = \nabla_\phi \int_{\statespace} \pi_\phi(x) J_\phi(x) \mathrm{d}x
\]
where $\pi_\phi(x) = Z(\phi)^{-1}\exp(W_\phi(x))$, $Z(\phi) = \int_\statespace \exp(W_\phi(x))\ud x$ and $J_\phi(x)$ is a function depending on $\phi$. Assuming we can interchange the gradient and the expectation and using the product rule we can rewrite:
\[
g_\phi(x)=\int_\statespace \left( J_\phi(x)\nabla_\phi\pi_\phi(x)+ \pi_\phi(x)\nabla_\phi J_\phi(x) \right) \mathrm{d}x.
\]
Using $\nabla_\phi \pi_\phi(x) = \pi_\phi(x)\nabla_\phi \log \pi_\phi(x)$,
\[
g_\phi(x)=\int_\statespace \pi_\phi(x) (J_\phi(x)\nabla_\phi\log\pi_\phi(x)+ \nabla_\phi J_\phi(x)) \mathrm{d}x.
\]
From the definition of $\pi_\phi(x)$, we can evaluate the score function as
\begin{align*}
\nabla_\phi \log \pi_\phi(x) &= -\nabla_\phi \log Z(\phi) + \nabla_\phi W_\phi(x)\\
&= -\E\left[\nabla_\phi W_{\phi}(x)\right] + \nabla_\phi W_\phi(x).
\end{align*}
Substitute this in $g_\phi(x)$ we obtain,
\begin{align*}
g_\phi(x)&= \int_\statespace \pi_\phi(x) J_\phi(x)(-\E\left[\nabla_\phi W_{\phi}(x)\right] + \nabla_\phi W_\phi(x))\mathrm{d}x  \\ 
& \quad + \int_\statespace \pi_\phi(x) \nabla_\phi J_\phi(x)\mathrm{d}x  \\
&=-\E[J_\phi(x)]\E[\nabla_\phi W_\phi(x)] + \E[J_\phi(x)\nabla_\phi W_\phi(x)] \\
& \quad + \E[\nabla_\phi J_\phi(x)]\\
&= \text{Cov}[\nabla_\phi W_\phi(x),J_\phi(x)]+ \E[\nabla_\phi J_\phi(x)].
\end{align*}

\subsection{Exponential family and linear function}\label{subsec:gradient_exponential_linear}
The gradient derived in the previous section can easily be applied to expectations with respect to functions linear in $\phi$ under distributions in the exponential family. Let $J_\phi(x) = \xi_J(\phi)^T J(x)$ be a linear function in $\phi$ and suppose $W_\phi(x) = \xi_W(\phi)^T W(x)$ for some functions $\xi_J : \R^d \to \R^n$, $J : \statespace \to \R^n$ and $\xi_W : \R^d \to \R^m$, $W : \statespace \to \R^m$. Then
\begin{align*}
g_\phi(x) &= \text{Cov}[\nabla_\phi W_\phi(x),J_\phi(x)]+ \E[\nabla_\phi J_\phi(x)]\\
&= \nabla_\phi \xi_W(\phi)^T \text{Cov}[W(x), J^T(x)]\xi_J(\phi) \\
& \quad + \nabla_\phi \xi_J(\phi)^T \E[J(x)]
\end{align*}
where $\nabla_\phi \xi(\phi)^T$ is the transposed Jacobian of $\xi$.

\subsection{Symmetric KL: general case}\label{subsec:SKL_general}
Next we show that the symmetric KL divergence of Equation {\ref{eq:bound}} can
be rewritten as a sum of expectations over functions parametrized by $\phi$,
hence falling in the framework presented above.

For path parameter $\phi$, the symmetric KL divergence is
\begin{align*}
\LSKL(\phi) &= \sum_{n=0}^{N-1} \text{SKL}(\pi_{t_{n}}^\phi, \pi_{t_{n+1}}^\phi) \\
&= \sum_{n=0}^{N-1} \E\left[\log\frac{\pi_{t_{n+1}}^\phi(X_{n+1})}{\pi_{t_n}^\phi(X_{n+1})} +  \log\frac{\pi_{t_{n}}^\phi(X_n)}{\pi_{t_{n+1}}^\phi(X_n)}\right]
\end{align*}
where $X_n\sim \pi_{t_{n}}^\phi$. After cancellation of the normalization constants we obtain
\begin{align*}
&\LSKL(\phi) =\\
& \sum_{n=0}^{N-1}\E\big[W_{t_{n+1}}^\phi(X_{n+1})-W_{t_n}^\phi(X_{n+1}) \\
&+ W_{t_n}^\phi(X_n) - W_{t_{n+1}}^\phi(X_n)\big].
\end{align*}
Collecting expectations under the same distribution and rearranging terms,
\begin{align*}
&\LSKL(\phi) =\\
&\E[W_{t_0}^\phi(X_0) - W_{t_1}^\phi(X_0)] +\\
&\sum_{n=1}^{N-1} \E[ 2W_{t_n}^\phi(X_n) - W_{t_{n+1}}^\phi(X_n) - W_{t_{n-1}}^\phi(X_n) ] + \\
& \E[W_{t_N}^\phi(X_N) - W_{t_{N-1}}^\phi(X_N)].
\end{align*}
Defining for $n=1, \dots, N-1$,
\begin{align*}
J_{0}^\phi(x) &= W_{t_0}^\phi(x) - W_{t_1}^\phi(x)\\
J_{n}^\phi(x) &= 2W_{t_n}^\phi(x) - W_{t_{n+1}}^\phi(x) - W_{t_{n-1}}^\phi(x)\\
J_{N}^\phi(x) &= W_{t_N}^\phi(x) - W_{t_{N-1}}^\phi(x),
\end{align*}
we have that
\[
\LSKL(\phi) = \sum_{n=0}^N \E[J_{n}^\phi(X_n)]
\]
and 
\[
\nabla_\phi \LSKL(\phi) = \sum_{n=0}^N \nabla_\phi\E[J_{n}^\phi(X_n)]
\]
where $\nabla_\phi \E[J_{n}^\phi(X_n)]$ can be computed using the formula derived in Section {\ref{subsec:gradient_expectation}}.

\subsection{Symmetric KL: exponential family case}\label{subsec:SKL_expfamily}
For the spline family introduce in Section~\ref{sec:expfamspline}, the distributions $\pi_{t_n}^\phi$ are in the exponential family with,
\[
W_{t_n}^\phi(x) = {\eta^\phi(t_n)}^T W(x), \quad n=0, \ldots, N.
\]
It follows that the functions $J_{n}^\phi$ are linear in $\phi$ with
\begin{align*}
J_{0}^\phi(x) &= {z^\phi_0}^TW(x)\\
J_{n}^\phi(x) &= {z^\phi_n}^TW(x), \quad n=1, \ldots, N - 1\\
J_{N}^\phi(x) &= {z^\phi_N}^TW(x),
\end{align*}
where
\begin{align*}
z^\phi_0 &= \eta^\phi(t_0) - \eta^\phi(t_1)\\
z^\phi_n &= 2\eta^\phi(t_n) - \eta^\phi(t_{n+1}) - \eta^\phi(t_{n-1}), \quad n=1, \ldots, N - 1\\
z^\phi_N &= \eta^\phi(t_N) - \eta^\phi(t_{N-1}).
\end{align*}
Given this relation, the stochastic gradient of Equation {\ref{eq:bound}} can be evaluated 
using $s$ samples from parallel tempering through the formula 
in Section {\ref{subsec:gradient_exponential_linear}} defining:
\begin{align*}
X &= (X_0, \ldots, X_N)\\
W(X) &= [W_0(X_0), W_1(X_0), \ldots, W_0(X_N), W_1(X_N)]^T\\
J(X) &= W(X) \\
\xi_W(\phi) &= [\eta^\phi_0(t_0), \eta^\phi_1(t_0), \ldots, \eta^\phi_0(t_N), \eta^\phi_1(t_N)]^T\\
\xi_J(\phi) &= [z^\phi_{0,0}, z^\phi_{0,1}, \ldots, z^\phi_{N,0}, z^\phi_{N,1}]^T 
\end{align*}
where $X$ is the $s \times N$ matrix of samples from parallel tempering, $W(X)$ is a $s \times 2N$ matrix evaluating $X$ elementwise at the reference and target distributions $W_0$ and $W_1$, $\xi_W(\phi)$ is a $2N \times 1$ vector of annealing coefficients and $\xi_J(\phi)$ is a $2N \times 1$ vector of coefficients defining $J^\phi = [J_0^\phi, \ldots, J_N^\phi]$.

\section{Proof of proposition 2}
For this annealing path family,
\[
W_t(x) = \eta(t)^TW(x).
\]
Therefore, the piecewise twice continuous differentiability of $\eta(t)$ and endpoint
conditions imply that Definition \ref{defn:path} is satisfied.
Next, note that if 
\begin{align*}
\sup_t \max\{\|\eta'(t)\|_2, \|\eta''(t)\|_2\} \leq M,
\end{align*}
then
\begin{align*}
\left|\frac{\dee W_t}{\dee t}\right| &= |\eta'(t)^TW(x)| \leq M\|W(x)\|_2\\
\left|\frac{\dee^2 W_t}{\dee t^2}\right| &= |\eta''(t)^TW(x)| \leq M\|W(x)\|_2,
\end{align*}
and thus by setting $V_1(x) = V_2(x) = M\|W(x)\|_2$ we satisfy
Equations (\ref{eq:V1_bound}) and (\ref{eq:second_deriv}).
Equation (\ref{eq:mgf}) implies Equation (\ref{eq:V1_integrability});
so as long as Equation (\ref{eq:mgf}) holds, the path $\eta$ satisfies 
all of the conditions of Theorem~\ref{thm:general_GCB}. 

Finally, note that $\Omega$ is a convex subset of $\R^2$:
for any nonnegative function $G(x)$,
vectors $\xi_1, \xi_2\in\R^2$, and $\lambda \in [0,1]$,
\begin{align*}
&\exp((\lambda \xi_1 + (1-\lambda)\xi_2)^TW(x)) G(x)\\
= &\left(\exp(\xi_1^TW)G(x)\right)^{\lambda}\left(\exp(\xi_1^TW)G(x)\right)^{1-\lambda}
\end{align*}
and so H\"older's inequality $\int f^\lambda g^{1-\lambda} \leq (\int f)^\lambda(\int g)^{1-\lambda}$
yields log-convexity (and hence convexity). Therefore as long as the endpoints $(0,1)$ and $(1,0)$
are both in $\Omega$, any convex combination of $(0,1)$ and $(1,0)$
is also in $\Omega$, and therefore the linear path $\eta(t) = (1-t, t)$ creates a set 
of normalizable densities and may be included in $\mathcal{A}$.

\section{Empirical support for the SKL surrogate objective function}\label{sec:snr}

\begin{figure*}
\centering
\includegraphics[width=0.4\linewidth]{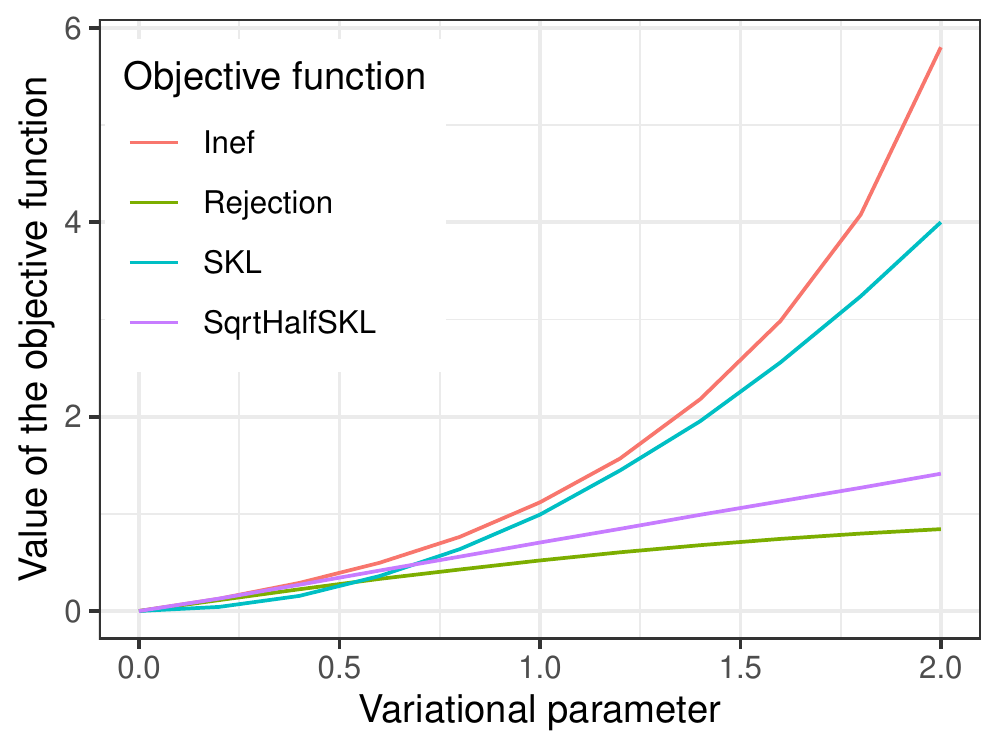}
\includegraphics[width=0.4\linewidth]{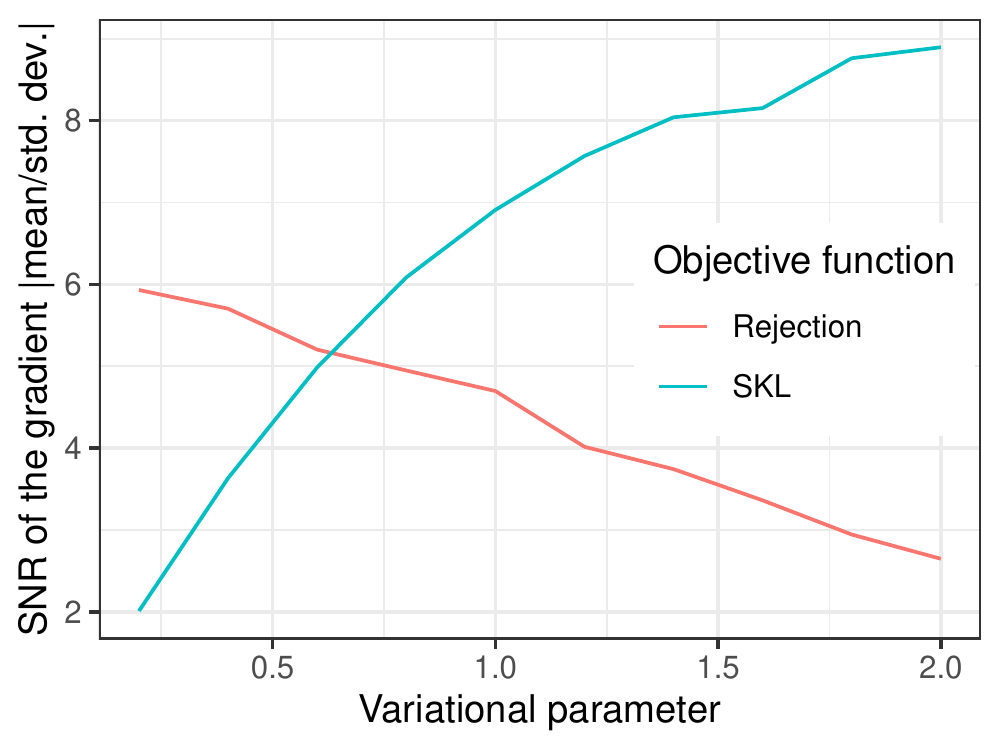}
\caption{Left: objective functions for path optimization in a controlled experiment as a function of a variational parameter $\phi$. Right: signal-to-noise of corresponding gradient estimators on the same range of parameters. }
\label{fig:objectives}
\end{figure*}

Two objective functions were discussed in Section~\ref{sec:general_annealing}: one based on rejection rate statistics, i.e.\ Equation~(\ref{eq:tauopt}), and the symmetric KL divergence (SKL). In this section we perform controlled experiments comparing the signal-to-noise ratio of Monte Carlo estimators of the gradient of these two objectives. 
Let $G$ denote a Monte Carlo estimator of a partial derivative with respect to one of the parameters $\phi_i$. Refer to {\ref{sec:gradient}} for details on the stochastic gradient estimators. In this experiment we use i.i.d.\ samples so that the Monte Carlo estimators are unbiased, justifying the use of the variance as a notion of noise. Hence following  \citet{RainforthKLMIWT18}, we define the signal-to-noise ratio by $\text{SNR} = |\E[G]/\sigma[G]|$, where $\sigma[G]$ denotes the standard deviation of $G$.
We use two chains with one set to a standard Gaussian, the other to a Gaussian with mean $\phi$ and unit variance. We show the value of the two objective functions in Figure~\ref{fig:objectives} (left). The label ``Rejection'' refers to the expected rejection of the swap proposal between the two chains, $r$. 
We also show the square root of half of the SKL (``SqrtHalfSKL''), to quantify the tightness of the bound in Equation~(\ref{eq:bound}), while ``Ineff'' shows the rejection odds, $r/(1-r)$, called inefficiency in \citet{syed_non_reversible_2019}. 

Signal-to-noise ratio estimates were computed for each parameter $\phi_i \in \{0, 1/5, 2/5, \dots, 2\}$. Each gradient estimate uses 50 samples, and to approximate the signal-to-noise ratio, the estimation was repeated 1000 times for each $\phi_i$ and objective function. The results are shown in Figure~\ref{fig:objectives} (right), and demonstrate that in the regime of small rejection ($\lessapprox 30\%$), the gradient estimator based on the rejection objective has a superior signal-to-noise ratio compared to its SKL counterpart. However as $\phi$ increases and the two distributions become farther apart, the situation is reversed, providing empirical support for the surrogate objective for challenging path optimization problems. 

\section{Experimental details}

All the experiments were conducted comparing reversible PT, non-reversible PT and non-reversible PT based on the spline family with $K \in \{2,3,4,5,10\}$. 

Every method was initialized at the linear path with equally spaced schedule, i.e. $\pi_t \propto \pi_0^{1 - t/N}\pi_1^{t/N}$ with $N$ the number of parallel chains. All methods performed one local exploration step before a communication step. 

To ensure a fair comparison of the different algorithms, we fixed the computational budget to a pre-determined number of samples in each experiment. Reversible PT used the budget to perform local exploration steps followed by communication steps. In non-reversible PT  the computational budget was used to tune the schedule according to the procedure described in \citet[Section~5.1]{syed_non_reversible_2019}. For non-reversible PT with path optimization, the computational budget was divided equally over a fixed number of scans of Algorithm \ref{alg:PT_tuning}, where a scan corresponds to one iteration of the for loop.

Optimization of the spline annealing path family was performed using the SKL surrogate objective of Equation {\ref{eq:bound}}. Adagrad was used for the optimization. The gradient was scaled elementwise by its absolute value
plus the value of the knot component. Such scaling was necessary to limit the gradient in the interval $[-1,1]$, stabilizing the optimization and avoiding possible exploding gradients due to the transformation to log space.

To mitigate variance in the results due to randomness, we performed 10 runs of each method and averaged the results across the runs.

\subsection{Gaussian}
This experiment optimized the path between the reference $\pi_0 = N(-1, 0.01^2)$ and the target $\pi_1 = N(1, 0.01^2)$. We used $N=50$ parallel chains initialized at a state sampled from a standard Gaussian distribution.
In this setting, $\pi_t$ has a closed form that can be shown to be $N\left(\frac{\eta_1(t)-\eta_0(t)}{\eta_0(t) + \eta_1(t)}, \left(\frac{0.01^2}{\eta_0(t) + \eta_1(t)}\right)^2\right)$, therefore, in the local exploration step of parallel tempering we sampled i.i.d.\  from $\pi_t$.
The computational budget was fixed at 45000 samples. Non-reversible PT with optimized path divided the budget in 150 scans. Therefore, for every gradient step in Algorithm \ref{alg:PT_tuning}, the gradient was estimated with 300 samples. We used 0.2 as learning rate for Adagrad.

\subsection{Beta-binomial model}
The second experiment was performed on a conjugate Bayesian model. The model prior was $\pi_0(p) = \mathrm{Beta}(180, 840)$. The likelihood was $L(x|p) = \mathrm{Binomial}(x|n, p)$. We simulated $x_1, \ldots, x_{2000} \sim \mathrm{Binomial}(100, 0.7)$, resulting in a posterior distribution $\pi_1(p) = \mathrm{Beta}(140180, 60840)$. The prior is concentrated at 0.176 with a standard deviation of 0.0119. The posterior distribution is concentrated at 0.697 with a standard deviation of 0.001.
We used $N=50$ parallel chains initialized at 0.5.
Also in this experiment it is possible to compute $\pi_t$ in closed form. Let $S = \sum_{i=1}^{2000} x_i$, $R = 2000 \times 100$ then $\pi_t(p)  = \mathrm{Beta}(179\eta_0(t) + (180 + S - 1)\eta_1(t) + 1, 839\eta_0(t) + (840 + N - S - 1)\eta_1(t) + 1)$. Hence, in the local exploration step of parallel tempering we sampled i.i.d.\ from $\pi_t$.
The computational budget was fixed at 45000 samples. Non-reversible PT with optimized path divided the budget in 150 scans. Therefore, for every gradient step in Algorithm \ref{alg:PT_tuning}, the gradient was estimated with 300 samples. We used 0.2 as learning rate for Adagrad.

\subsection{Galaxy data}
The third experiment was a Bayesian Gaussian mixture model applied to the galaxy dataset of \citet{roeder1990density}. We used six mixture components with mixture proportions $w_0, \ldots, w_5$, mixture component densities $N(\mu_i, 1)$ for mean parameters $\mu_0, \ldots, \mu_5$, and a binary cluster label for each data point. We placed a $\mathrm{Dir}(\boldsymbol{1})$ prior on the
proportions, where $\boldsymbol{1} = (1,1,1,1,1,1)$ and a $N(150, 1)$ prior on each of the mean parameters. We did not marginalize the cluster indicators, creating a multi-modal posterior inference problem over 94 latent variables. In this experiment we used $N=35$ chains. Mixture proportions were initialized at $1/6$, mean parameters were initialized at 0 and cluster labels were initialized at 0. The local exploration step involved standard Gibbs steps for the means, indicators,
and proportions. To improve local mixing, we also included an additional Metropolis-Hastings 
step for the proportions that approximates a Gibbs step when the indicators are marginalized. We fixed the computational budget to 50000 samples, divided into 500 scans using 100 samples each. We optimized the path using Adagrad with a learning rate of 0.3.

\subsection{Mixture model}
The fourth experiment was a Bayesian Gaussian mixture model with mixture 
proportions $w_0, w_1$, mixture component densities $N(\mu_i, 10^2)$
for mean parameters $\mu_0,\mu_1$, and a binary cluster label for each data point.
We placed a $\mathrm{Dir}(1, 1)$ prior on the
proportions, and a $N(150, 1)$ prior on each of the two mean parameters. 
We simulated $n=1000$ data points from the mixture $0.3N(100, 10^2) + 0.7N(200,
10^2)$. We did not marginalize the cluster indicators, creating a multi-modal posterior over 1004 latent variables.
We used $N=35$ chains. Mixture proportions were initialized at 0.5, mean
parameters were initialized at 0 and cluster labels were initialized at 0.
The local exploration step involved standard Gibbs steps for the means, indicator variables,
and proportions. To improve local mixing, we also included an additional Metropolis-Hastings 
step for the proportions that approximates a Gibbs step when the indicators are marginalized.
The computational budget was fixed at 25000 samples. Non-reversible PT with
optimized path divided the budget in 50 scans. Therefore, for every gradient
step in Algorithm \ref{alg:PT_tuning}, the gradient was estimated with 500
samples. We used 0.3 as learning rate for Adagrad. Results are shown in Figure \ref{fig:results_mixture}.

\begin{figure}[t!]
\vskip -3.8in
\begin{center}
\centerline{
\includegraphics[width=\columnwidth]{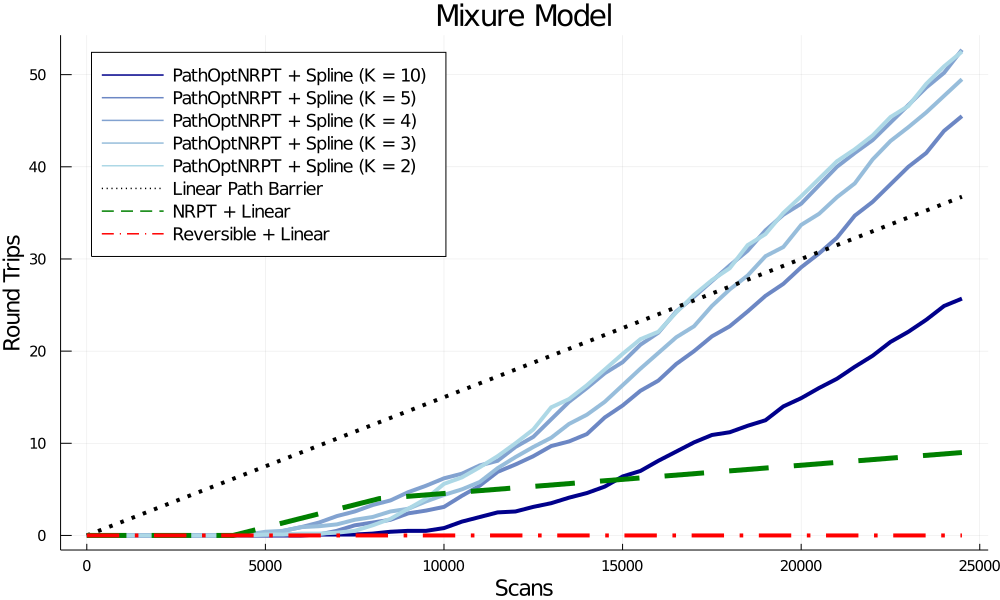}
}
\centerline{
\includegraphics[width=\columnwidth]{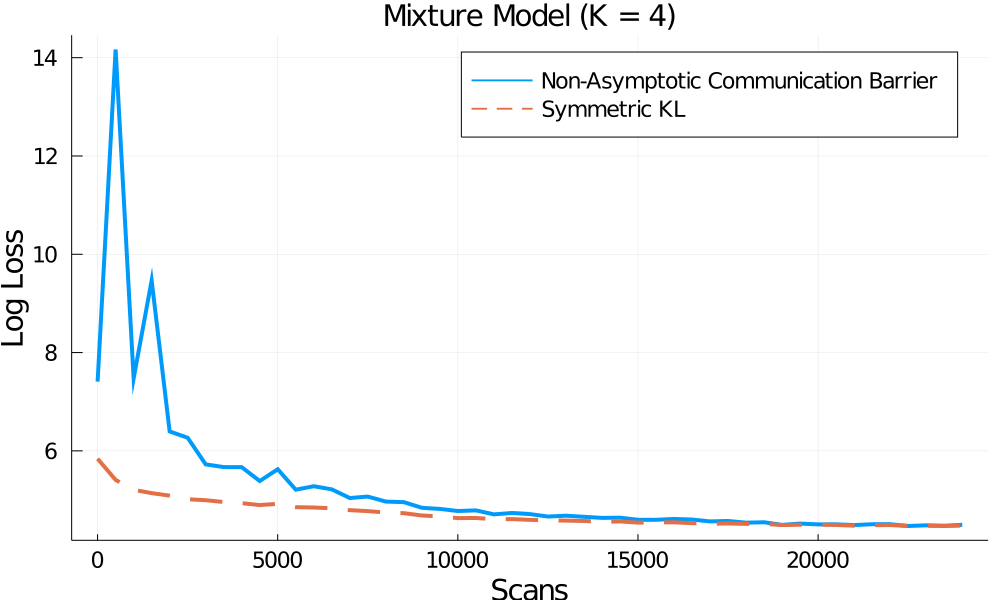}
}
\caption{\textbf{Top:} Cumulative round trips averaged over 10 runs for the spline path with $K=2,3,4,5,10$ (solid blue), NRPT using a linear path (dashed green), and reversible PT with linear path (dash/dot red). The slope of the lines represent the round trip rate. \textbf{Bottom:} Non-asymptotic communication barrier from Equation \ref{eq:nonasympcb} (solid blue) and Symmetric KL (dash orange) as a function of iteration for one run of PathOptNRPT + Spline ($K=4$ knots).}
\label{fig:results_mixture}
\end{center}
\vskip -0.2in
\end{figure}

\end{document}